\documentclass[a4paper]{article}

\usepackage{amsthm}
\usepackage{amssymb}
\usepackage{amsmath}
\usepackage{stmaryrd}
\usepackage{pifont}

\usepackage{tikz}
\usetikzlibrary{arrows,automata}

\usepackage{hyperref}
\hypersetup{bookmarks = true, colorlinks = false}

\usepackage{url}

\usepackage
{geometry}

\usepackage{color}
\usepackage{gb4e}
\noautomath

\usepackage[edges]{forest}
\usetikzlibrary{arrows.meta}

\theoremstyle{plain}
\newtheorem{definition}{Definition}

\theoremstyle{plain}
\newtheorem{lemma}{Lemma}

\theoremstyle{plain}
\newtheorem{theorem}{Theorem}

\theoremstyle{plain}
\newtheorem{example}{Example}

\theoremstyle{plain}
\newtheorem{fact}{Fact}

\theoremstyle{plain}
\newtheorem{rfact}{Fact}

\theoremstyle{plain}
\newtheorem{rlemma}{Lemma}

\theoremstyle{plain}
\newtheorem{rtheorem}{Theorem}

\newcommand{\PL}{\mathsf{PL}}

\newcommand{\ConSHNBT}{\mathsf{ConSHN\text{-}BT}}
\newcommand{\XY}{\mathsf{XY}}

\newcommand{\ConXXYY}{\mathsf{Con}\text{-}\vec{\mathsf{X}}\vec{\mathsf{Y}}}
\newcommand{\NXXYY}{\mathsf{N}\text{-}\vec{\mathsf{X}}\vec{\mathsf{Y}}}
\newcommand{\OneConXXYY}{\mathsf{OneCon}\text{-}\vec{\mathsf{X}}\vec{\mathsf{Y}}}
\newcommand{\OneBXXYY}{\mathsf{One}\Box\text{-}\vec{\mathsf{X}}\vec{\mathsf{Y}}}

\newcommand{\AP}{\mathsf{AP}}

\newcommand{\XXX}{\mathtt{X}}
\newcommand{\YYY}{\mathtt{Y}}

\newcommand{\con}[1]{[#1]}

\newcommand{\dcon}[1]{\langle #1 \rangle}

\newcommand{\MM}{\mathtt{M}}

\newcommand{\TL}[1]{\mathtt{TL}(#1)}

\newcommand{\CC}{\mathtt{C}}

\newcommand{\rul}{R}

\newcommand{\AT}[1]{\mathtt{AT}(#1)}

\newcommand{\EN}[1]{\bigcap #1}

\newcommand{\upl}[2]{#1 \cup \{#2\}}

\newcommand{\golf}[2]{|#1|_{#2}}

\newcommand{\upf}[3]{#1 +_{#3} #2}

\newcommand{\exs}[1]{\textit{#1}}

\newcommand{\ecs}[1]{{\color{blue}(#1)}}

\newcommand{\red}[1]{{\color{red}#1}}

\newcommand{\fcut}[1]{}

\newcommand{\bee}{\begin{exe}}
\newcommand{\eee}{\end{exe}}

\newcommand{\bi}{\begin{itemize}}
\newcommand{\ei}{\end{itemize}}

\newcommand{\defs}{\textbf}

\title{Logic for conditional strong historical necessity in branching time and analyses of an argument for future determinism}
\author{
Fengkui Ju \smallskip \\
{\small School of Philosophy, Beijing Normal University, Beijing, China} \\
{\small \emph{fengkui.ju@bnu.edu.cn}}
}
\date{}

\begin{document}

\maketitle

\begin{abstract}
\noindent In this paper, we present a logic for conditional strong historical necessity in branching time and apply it to analyze a nontheological version of Lavenham's argument for future determinism. Strong historical necessity is motivated from a linguistical perspective, and an example of it is ``If I had not gotten away, I must have been dead''.
The approach of the logic is as follows. The agent accepts ontic rules concerning how the world evolves over time. She takes some rules as indefeasible, which determine acceptable timelines. When evaluating a sentence with conditional strong historical necessity, we introduce its antecedent as an indefeasible ontic rule and then check whether its consequent holds for all acceptable timelines.
The argument is not sound by the logic.

\medskip

\noindent \textbf{Keywords:} indefeasible ontic rules; contexts; acceptable timelines; update; a nontheological version of Lavenham's argument for future determinism
\end{abstract}

\section{Introduction}
\label{sec:Introduction}

\emph{Necessities} are universal quantifications over domains of possibilities. There are many necessities in the literature: \emph{logical}, \emph{metaphysical}, \emph{natural}, \emph{deontic}, \emph{epistemic}, and so forth. These notions are studied in different contexts: logical, philosophical, ethical, legal, linguistical, and so forth. These notions may have different meanings in different works. We refer to \cite{sep-modality-varieties} and \cite{portner_modality_2009} for general discussions about them.
In this paper, we propose a logic for a conditional necessity in branching time, called \emph{conditional strong historical necessity}, and use the logic to analyze an argument for future determinism.

The necessity is motivated mainly from a linguistical perspective.
The argument, from \cite{ohrstrom_future_2020}, is a nontheological version of an argument given by the scholastic philosopher Richard of Lavenham in his \emph{De eventu futurorum}. Lavenham's argument is close to Aristotle's sea battle puzzle, given in his \emph{On Interpretation}. By \cite{ohrstrom_future_2020}, \cite{ohrstrom_temporal_1995} and \cite{jarmuzek_sea_2018}, it is also close to the Master argument, given by the Greek philosopher Diodorus Cronus.

\subsection{Strong historical necessity}
\label{subsec:Strong historical necessity}

\emph{Historical possibilities} are possible states that \emph{we think} our world is in, or could be in, if things had gone differently in the past. Historical possibilities can be easily confused with \emph{epistemic possibilities}. An epistemic possibility is a possible state that we think is possible to be the actual state of the world. As an example, suppose Adam bought a lottery ticket yesterday, the winning number was just revealed, and Adam loses. Assume Bob knows that Adam loses. Then, the possibility that Adam wins is historical but not epistemic for Bob. Assume Bob does not know that Adam loses. Then, the possibility that Adam wins is both historical and epistemic for Bob.

Epistemic possibilities are relative to agents: it can happen that one agent thinks that a possible state might be the actual state of the world, but another agent does not think so. Historical possibilities are relative to agents, too. For example, people often disagree about whether an ancient dynasty could have a different fate than its actual fate.

Strong historical necessity is the historical necessity that can be expressed by ``must'' and ``necessary'' in English. Its dual, the \emph{strong historical possibility}, can be expressed by ``might'' and ``could''. In Chinese, strong historical necessity can be expressed by ``biding'', ``yiding'', and ``biran''.

We consider some examples of strong historical necessity:

\begin{exe}
\ex A kindergarten organizes a lottery game. For some reason, Adam and Bob exchange their tickets. In the end, Adam wins. After the draw, the kindergarten leader says the following: \par \vspace{2pt} \exs{The winner could be anyone else and it is not that Adam \textbf{must} be the winner.} \label{ex:Adam must be the winner}
\ex \exs{Our guests are home by now. In fact, they \textbf{must} be home by now. They left half-an-hour ago, have a fast car, and live only a few miles away.} \label{ex:they must be home by now} \ecs{Adapted from \cite{leech_meaning_1971}}
\ex \exs{My favorite number is \textbf{necessarily} even.} \label{ex:My favorite number is necessarily even} \ecs{From \cite{thomasson_norms_2020}}
\end{exe}

\noindent In the scenario of \ref{ex:Adam must be the winner}, it is certain that Adam is the winner. However, it is still felicitous to negate ``Adam must be the winner''. This means that ``must'' in \ref{ex:Adam must be the winner} is not epistemic. ``Must'' in \ref{ex:they must be home by now} is not epistemic either; otherwise, the sentence ``they must be home by now'' would be redundant. Suppose an accidental past event made 14 my favorite number. Then, there is a sense in which \ref{ex:My favorite number is necessarily even} is false. This means that ``necessarily'' in \ref{ex:My favorite number is necessarily even} can be historical.

Some examples of strong historical possibility follow:

\begin{exe}
\ex \exs{He \textbf{might} have won the game, but he didn't in the end.} \label{ex:He might have won the game} \ecs{Adapted from \cite{portner_modality_2009}}
\ex \exs{I \textbf{could} not have done better.} \label{ex:I could not have done better} \ecs{From \cite{palmer_modality_1990}}
\ex \exs{No one \textbf{could} have foreseen how long his illness was going to go on.} \label{ex:No one could have foreseen how long his illness was going to go on}
\end{exe}

\noindent By \ref{ex:He might have won the game}, it is epistemically settled that he did not win. So, ``might'' in it is not epistemic. Note that \ref{ex:I could not have done better} has the same meaning as ``It is not the case that I could have done better than what I have''. Then, we can see that ``could'' in \ref{ex:I could not have done better} is not epistemic; otherwise, it would be a trivial sentence. Typically, \ref{ex:No one could have foreseen how long his illness was going to go on} is uttered in a situation where it is clear that no one has foreseen how long his illness was going to go on. In this situation, \ref{ex:No one could have foreseen how long his illness was going to go on} would be trivial if we understand ``could'' in it as epistemic. This implies that ``could'' in it is not epistemic.

\paragraph{Remarks}

Strong historical necessity gets its name against another historical necessity, called \emph{weak} historical necessity, which can be expressed by ``should'' and ``ought to'' in English and by ``yingdang'' and ``yinggai'' in Chinese.

Here are some examples of weak historical necessity:

\begin{exe}
\ex Suppose Jones is in a building when an earthquake hits. The building collapses. Luckily, nothing falls upon Jones and he emerges from the rubble as the only survivor. Talking to the media, Jones says the following: \par \vspace{2pt} \exs{I \textbf{ought to} be dead right now.} \label{ex:I ought to be dead right now} \ecs{From \cite{yalcin_modalities_2016}}
\ex Consider Rasputin. He was hard to kill. First his assassins poisoned him, then they shot him, then they finally drowned him. Let us imagine that we were there. Let us suppose that the assassins fed him pastries dosed with a powerful, fast-acting poison, and then left him alone for a while, telling him they would be back in half an hour. Half an hour later, one of the assassins said to the others, confidently, ``\exs{He ought to be dead by now.}'' The others agreed, and they went to look. Rasputin opened his eyes and glared at them. ``\exs{He \textbf{ought to} be dead by now!}'' they said, astonished. \label{ex:Consider Rasputin} \ecs{From \cite{thomson_normativity_2008}}
\ex \exs{The beer \textbf{should} be cold by now, but it isn't.}
\end{exe}

\noindent The sentence \ref{ex:I ought to be dead right now} is felicitous, but its prejacent is clearly false. This implies that ``ought to'' in it is not epistemic. This is the same for the second occurrence of ``ought to'' in \ref{ex:Consider Rasputin}.

In general, necessities expressed by ``should'' and ``must'' have different properties. We refer to \cite{mcnamara_must_1996}, \cite{copley_what_2006} and \cite{fintel_how_2008} for detailed discussion. Here, we mention three differences between them. First, ``should $\phi$'' is weaker than ``must $\phi$''. Second, ``should $\phi$'' can coexist with ``not $\phi$'', but ``must $\phi$'' cannot. Third, ``should $\phi$'' is gradable, but ``must $\phi$'' is not.

The notion of strong historical necessity discussed here differs from the notion of historical necessity discussed by Thomason \cite{thomason_combinations_1984}. We call it the \emph{futural necessity} in the sequel. The main difference between the two notions can shown as follows: if an atomic proposition is true now, then it is necessary for futural necessity, but this is not the case for strong historical necessity.

\subsection{Conditional strong historical necessity}
\label{subsec:Conditional strong historical necessity}

\emph{Conditional strong historical necessity} states that strong historical necessity holds under a proposed situation.

What follows is some examples of conditional strong historical necessity:

\bee
\ex \exs{If he had stayed in the army, he \textbf{must} have become a colonel.} \label{ex:If he had stayed in the army must} \ecs{From \cite{palmer_modality_1990}}
\ex \exs{If she had taken the train, she \textbf{must} have arrived sooner.} \label{ex:If she had taken the train} \ecs{Adapted from \cite{portner_modality_2009}}
\ex \exs{If that had been a raccoon you heard in the garden last night, there \textbf{must} be its tracks in the snow now.} \label{ex:If that had been a raccoon you heard in the garden last night}
\eee

Here are some examples of \emph{conditional strong historical possibility}:

\bee
\ex \exs{If he had stayed in the army, he \textbf{might} have become a colonel.} \ecs{From \cite{palmer_modality_1990}}
\ex \exs{This \textbf{might} not have mattered much had the November declaration been seen as having put sterling wholly out of danger.} \ecs{From \cite{palmer_modality_1990}}
\ex \exs{That \textbf{could} not have happened if I had been killed.} \label{ex:if I had been killed} \ecs{Adapted from \cite{palmer_mood_2001}}
\eee

Conditional strong historical necessity is not trivial to address. Consider the well-known fable ``\emph{The lady, or the tiger?}'' from the American writer Frank Stockton \cite{wikisource_lady_2015}. Suppose it is the dramatic moment when the princess is going to point at a door. It seems correct to say \exs{if she points at the door with a lady behind it, she will necessarily be jealous, and if she points at the door with a tiger behind it, she will necessarily be sorry}. This seems to imply \emph{she will necessarily be jealous or she will necessarily be sorry}. However, it seems incorrect to say \exs{she will necessarily be jealous}, and it seems incorrect to say \exs{she will necessarily be sorry}.

\paragraph{Remarks}

It is commonly accepted that general conditionals have an implicit modality in the consequent~\cite{kratzer_modality_1991,stalnaker_defense_1981}, although it is a controversial issue what the implicit modality is in general conditionals. Are some general conditionals conditional strong historical necessity? We leave a deep investigation of this interesting issue for future work. Here, we mention one thing which seems to imply a negative answer. Consider the following two sentences from \cite{wawer_towards_2015}:

\bee
\ex \exs{If I had flipped the coin, it would have landed heads.}
\ex \exs{If I had flipped the coin, it would have \textbf{necessarily} landed heads.}
\eee

\noindent Wawer and Wro\'{n}ski \cite{wawer_towards_2015} think that our belief degree about the first sentence is $0.5$ but our belief degree about the second one is $0$.

\subsection{Related works on strong historical necessity and conditional strong historical necessity}
\label{subsec:Related works on strong historical necessity and conditional strong historical necessity}

Ju \cite{ju_logical_2022} presented a formal theory for strong and weak historical necessities in branching time. Roughly, the ideas for the two notions are as follows: the domains for strong and weak historical necessities, respectively, consist of acceptable timelines determined by indefeasible ontic rules, and expectable historical timelines determined by ontic rules.
There are many works on the futural necessity, such as \cite{prior_past_1967} and \cite{aqvist_logic_1999}.

We do not know of any works in the literature explicitly handling conditional strong historical necessity. Some works explicitly or implicitly deal with conditional futural necessity.
Ju, Grilletti and Goranko \cite{ju_logic_2018} presented a logic $\mathsf{LTC}$ to formalize reasoning about temporal conditionals ``\exs{if $\phi$, it is necessary that $\psi$}'', where both $\phi$ and $\psi$ concern the future.
Thomason and Gupta \cite{thomason_theory_1980}, Nute \cite{nute_historical_1991}, Placek and M\"{u}ller \cite{placek_counterfactuals_2007}, and Wawer and Wro\'{n}ski \cite{wawer_towards_2015} studied general conditionals with temporal dimension. In their semantics for general conditionals, there is an implicit universal quantifier in the consequent, which can be viewed as the futural necessity.

\subsection{Our work}
\label{subsec:Our work}

The rest of the paper is structured as follows. In Section \ref{sec:Logic for conditional strong historical necessity in branching time}, we present a logic for conditional strong historical necessity in branching time. In Section \ref{sec:Analyses of the nontheological version of Lavenham's argument}, we analyze the nontheological version of Lavenham's argument and compare our analyses to some related works. The paper concludes in Section \ref{sec:Conclusion}. We include proofs for some of the results in the Appendix.

Our understanding of strong historical necessity follows \cite{ju_logical_2022}. Our way to address conditional necessity follows the approach of the update semantics for conditionals proposed by Veltman \cite{veltman_making_2005}: when evaluating a conditional, we update something with its antecedent and then evaluate its consequent with respect to the update result. By our logic, the argument is not sound.

\section{Logic for conditional strong historical necessity in branching time}
\label{sec:Logic for conditional strong historical necessity in branching time}

\subsection{Our approach}
\label{subsec:Our approach}

The world can evolve in different ways in the ontic sense and there are different timelines. The agent accepts a set of \emph{ontic rules} concerning which timelines are acceptable. These ontic rules can be of many kinds: natural laws (\emph{light is faster than sound}), common natural phenomena (\emph{it is cold in Beijing during the winter}), daily experiences (\emph{lack of engine oil causes damage to engines}), and so on.

These rules may conflict with each other, and some may overrule others. Which rules defeat which others is determined by many factors. Here, we mention two of them. First, some types of rules tend to defeat other types. For example, natural laws often override daily experiences. Second, special rules tend to defeat general ones. For example, consider the following two rules: (1) \emph{the higher the altitude is, the colder the weather is}; (2) \emph{in winter in Yili Valley, the higher the altitude is, the warmer the weather is}. The latter is a special rule relative to the former one. When conflicting, the latter tends to override the former.

The agent might learn new ontic rules, and she is ready for some new ones to defeat some old ones. However, she might take some rules as \emph{indefeasible}. For example, we usually consider the following rule indefeasible: \emph{the sun will rise tomorrow}. The timelines violating any indefeasible ontic rules are \emph{unacceptable} for the agent.

\paragraph{Remarks} What indefeasible ontic rules agents accept is a complicated issue. In addition, different agents might accept different indefeasible ontic rules.

\medskip

The sentence ``\exs{if $\phi$ is the case, $\psi$ must be the case}'' is true at a moment with respect to a set of indefeasible ontic rules if and only if $\psi$ is true for all acceptable timelines determined by the result of adding $\phi$ to the set as an indefeasible ontic rule. Note that treating $\phi$ as an indefeasible ontic rule is hypothetic.

We analyze an example by these ideas. This will be our running example.

\begin{example} \label{example:tiger}
Some tribesmen capture a man called Adam and put him in a cave with a door. There are two small rooms and a leashed lion in the cave, which is close to the cave door and is not in either room. The chief tells Adam the following. At midnight, the lion will be unleashed. Adam can enter either room, but once he enters a room, its door will be locked. There is a tiger in one room and nothing in the other. If Adam is still alive tomorrow, then he will be released. In fact, the tiger is in the left room, but the chief does not tell Adam this. Figure \ref{fig:tiger} indicates the situation. After locking the cave door, the chief says the following:

\bee
\ex \exs{Adam will \textbf{necessarily} be in a room tomorrow. If he is in the left one, then he will \textbf{necessarily} be dead. If he is in the right one, then he will \textbf{necessarily} be alive.}\label{ex:Adam must be in a room tomorrow}
\eee

\noindent The next morning, the tribesmen come outside of the cave. The chief says the following:

\bee
\ex \exs{Adam is \textbf{necessarily} in a room now. If he is in the left one, then he \textbf{must} have been dead. If he is in the right one, then he \textbf{must} still be alive.} \label{ex:Adam must be in a room now 1}
\eee

\noindent They enter the cave and find Adam in the right room, alive. The chief says the following:

\bee
\ex \exs{Adam is \textbf{necessarily} in a room now. If he were in the left one, then he \textbf{must} have been dead. Since he is in the right one, he \textbf{must} still be alive.} \label{ex:Adam must be in a room now 2}
\eee

\begin{figure}[h]
\begin{center}
\begin{tikzpicture}[node distance=10mm,scale=0.80,every node/.style={transform shape}]

\draw[rounded corners] (-40mm,-20mm) rectangle (40mm,20mm);
\draw[rounded corners] (-40mm,5mm) rectangle (-15mm,20mm);
\draw[rounded corners] (15mm,5mm) rectangle (40mm,20mm);

\node[] (s-0) [] {};
\node[above=10mm,rounded corners,fill=gray!10,draw] (s-1) at (s-0) {\textsc{ADAM}};
\node[below=10mm,rounded corners,fill=gray!10,draw] (s-2) at (s-0) {\textsc{LION}};
\node [left=20mm,rounded corners,fill=gray!10,draw] (s-3) at (s-1) {\textsc{TIGER}};

\put (-38,34){\rotatebox{90}{$\shortmid$}}
\put (-38,21){\rotatebox{90}{$\shortmid$}}

\put (30.5,34){\rotatebox{90}{$\shortmid$}}
\put (30.5,21){\rotatebox{90}{$\shortmid$}}

\put (-10,-47.5){$\shortmid$}
\put (10,-47.5){$\shortmid$}

\end{tikzpicture}
\end{center}

\caption{This figure indicates the situation of Example \ref{example:tiger}.}
\label{fig:tiger}
\end{figure}
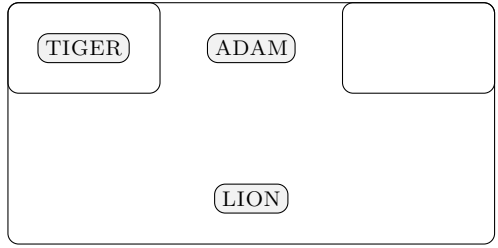
\end{example}

It seems reasonable to assume there are four indefeasible ontic rules in this example: (1) \emph{lions kill people}, (2) \emph{people want to survive}, (3) \emph{tigers kill people}, and (4) \emph{people die for a reason.} Note that Adam has just three choices: \emph{staying outside of the two rooms}, \emph{going into the left room} and \emph{going into the right room}.

Consider the first situation. By the second rule, it is unacceptable that \emph{Adam will not be in a room tomorrow}. Then, \emph{Adam will necessarily be in a room tomorrow}. Suppose we introduce a new indefeasible ontic rule: \emph{Adam will be in the left room tomorrow}. By the introduced rule and the third rule, it is unacceptable that \emph{Adam will be alive tomorrow}. Then, \emph{if Adam is in the left room tomorrow, then he will necessarily be dead}. Suppose we introduce a new indefeasible ontic rule: \emph{Adam will be in the right room tomorrow}. By the introduced rule and the fourth rule, it is unacceptable that \emph{Adam will be dead tomorrow}. Then, \emph{if Adam is in the right room tomorrow, then he will necessarily be alive}.

Consider the second situation. By the second rule, it is unacceptable that \emph{Adam is not in a room now}. Then, \emph{Adam is necessarily in a room now}. Suppose we introduce a new indefeasible ontic rule: \emph{Adam is in the left room now}. By the introduced rule and the third rule, it is unacceptable that \emph{Adam is alive now}. Then, \emph{if Adam is in the left room now, then he must have been dead}. Suppose we introduce a new indefeasible ontic rule: \emph{Adam is in the right room now}. By the introduced rule and the fourth rule, it is unacceptable that \emph{Adam is dead now}. Then, \emph{if Adam is in the right room now, then he must still be alive}.

The third situation is similar to the second one. The difference is that those tribesmen know Adam is alive in the right room in the third situation. However, whether they know this does not affect whether they accept the propositions expressed by the three sentences. It just affects their sentences: they would use subjunctive mood for the second sentence and ``since'' instead of ``if'' for the third sentence.

\subsection{Languages}
\label{subsec:Languages}

\begin{definition}[The languages $\Phi_{\XY}$ and $\Phi_{\ConSHNBT}$] \label{def:The languages Phi XY and Phi ConSONBT}
Let $\AP$ be a countable set of atomic propositions and $p$ range over it. The language $\Phi_{\XY}$ is defined as follows:

\[\alpha ::= p \mid \bot \mid \neg \alpha \mid (\alpha \land \alpha) \mid \XXX \alpha \mid \YYY \alpha\]

The language $\Phi_{\ConSHNBT}$ of the Logic for Conditional Strong Historical Necessity in Branching Time ($\ConSHNBT$) is defined as follows, where $\alpha \in \Phi_{\XY}$:

\[\phi ::= p \mid \bot \mid \neg \phi \mid (\phi \land \phi) \mid \XXX \phi \mid \YYY \phi \mid \con{\alpha} \phi\]
\end{definition}

The intuitive reading of featured formulas of $\Phi_{\XY}$ and $\Phi_{\ConSHNBT}$ is as follows:

\bi
\item $\XXX \phi$: \emph{$\phi$ will be true at the next instant.}
\item $\YYY \phi$: \emph{$\phi$ was true at the last instant.}
\item $\con{\alpha} \phi$: \emph{if $\alpha$, $\phi$ must be true at the present instant.}
\ei

What follow are some derivative expressions:

\bi
\item The propositional connectives $\top, \lor, \rightarrow$ and $\leftrightarrow$ are defined as usual.
\item Define the dual $\dcon{\alpha} \phi$ of $\con{\alpha} \phi$ as $\neg \con{\alpha} \neg \phi$, meaning \emph{if $\alpha$, $\phi$ could be true at the present instant}.
\item Define $\Box \phi$ as $\con{\top} \phi$, meaning \emph{$\phi$ must be true at the present instant.} Define the dual $\Diamond \phi$ of $\Box \phi$ as $\neg \Box \neg \phi$, indicating \emph{$\phi$ could be true at the present instant.}
\ei

\subsection{Semantic settings}
\label{subsec:Semantic settings}

\subsubsection{Models}
\label{subsubsec:Models}

\begin{definition}[Models for $\Phi_{\ConSHNBT}$] \label{def:Models for Phi ConSONBT}
A tuple $\MM = (W, <, V)$ is a \defs{model} for $\Phi_{\ConSHNBT}$ if the following conditions are satisfied:

\bi
\item $W$ is a non-empty set of states;
\item $<$ is a serial relation on $W$ meeting the following condition: there is a $r \in W$, called the root, such that for every $w \in W$, there is a unique finite sequence $x_0, \dots, x_n$ of states in $W$ such that $x_0 = r$, $x_n = w$ and $x_0 < \dots < x_n$;
\item $V: \mathsf{AP} \to \mathcal{P}(W)$ is a valuation.
\ei
\end{definition}

Intuitively, every element of $W$ represents a possible state of the world at an instant; $x < y$ indicates that the state $x$ can evolve to the state $y$ at the next instant; models indicate how the world can evolve in the time flow: there is a starting point but no ending point; the past is determined, but the future is open.

What follows are some notions related to models, which will be used later. Fix a model $\MM$. An infinite sequence $x_0, x_1, \dots$ of states is called a \emph{timeline} if $x_0 = r$ and $x_0 < x_1 < \dots$. We use $\TL{\MM}$ to indicate the set of all timelines of $\MM$. For any timeline $\pi$ and natural number $i$, we use $\pi[i]$ to indicate the $i+1$-th element of $\pi$. We call natural numbers \emph{instants}.

For every model $\MM$, timeline $\pi$ and instant $i$, $(\MM, \pi, i)$ is called a \defs{pointed model}.

\begin{example}[Pointed models] \label{example:Pointed models}
Figure \ref{fig:a pointed model} indicates a pointed model $(\MM,\pi_2,1)$. The intuitive reading of it is as follows. The timeline $\pi_2$ is in consideration. The present state is $\pi_2[1]$, that is, $w^1_1$, which has an alternative $w^1_2$. Both $w^1_1$ and $w^1_2$ evolve from the root $w^0_1$, each having two successors.

\begin{figure}[h]
\begin{center}
\begin{tikzpicture}[node distance=20mm,scale=0.8,every node/.style={transform shape}]
\tikzstyle{every state}=[draw=black,text=black,minimum size=15mm]

\node[state] (s-0-1) [] {$w^0_1 \atop \{p,q\}$};
\node[] (a) [above of=s-0-1] {};
\node[state] (s-1-1) [left of=a] {$w^1_1 \atop \{p\}$};
\node[state] (s-1-2) [right of=a] {$w^1_2 \atop \{q\}$};

\node[] (b) [above of=a] {};
\node[state] (s-2-2) [left of=b] {$w^2_2 \atop \{p\}$};
\node[state] (s-2-1) [left of=s-2-2] {$w^2_1 \atop \{\}$};
\node[state] (s-2-3) [right of=b] {$w^2_3 \atop \{q\}$};
\node[state] (s-2-4) [right of=s-2-3] {$w^2_4 \atop \{\}$};

\node [above=9mm] at (s-2-1) {$\pi_1 \atop \vdots$};
\node [above=9mm] at (s-2-2) {$\pi_2 \atop \red{\vdots}$};
\node [above=9mm] at (s-2-3) {$\pi_3 \atop \vdots$};
\node [above=9mm] at (s-2-4) {$\pi_4 \atop \vdots$};

\path
(s-0-1) edge [->,left,color=red] node {} (s-1-1)
(s-0-1) edge [->,left] node {} (s-1-2)

(s-1-1) edge [->,left] node {} (s-2-1)
(s-1-1) edge [->,left,color=red] node {} (s-2-2)

(s-1-2) edge [->,left] node {} (s-2-3)
(s-1-2) edge [->,left] node {} (s-2-4)
;

\node[] (2) [left of=s-2-1] {$2$};
\node[] (1) [below of=2] {\hspace{-18pt} \ding{248} $1$};
\node[] (0) [below of=1] {$0$};

\node [above=9mm] at (2) {$\textit{Instants} \atop \vdots$};
\end{tikzpicture}
\end{center}

\caption{A pointed model} \label{fig:a pointed model}
\end{figure}
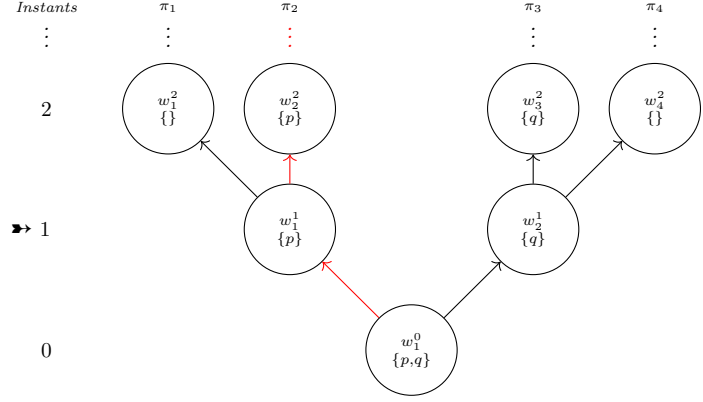
\end{example}

\subsubsection{Contexts}
\label{subsubsec:Contexts}

\begin{definition}[Contexts] \label{def:Contexts}
Let $\MM$ be a model. A finite (possibly empty) set $\CC$ of (possibly empty) sets of timelines is called a \defs{context} for $\MM$. The elements of $\CC$ are called \defs{indefeasible ontic rules}.
\end{definition}

\noindent In the sequel, we simply call indefeasible ontic rules \emph{rules}.

A set of timelines represents an ontic rule in the following sense: the timelines in the set are \emph{enabled} by the rule. We look at an example. The ontic rule \emph{super solar storms interfere with communications in the earth} excludes those timelines where there is a super solar storm at some point, but communications in the earth are not affected later. This ontic rule can be indicated as the set of timelines not excluded.

\begin{definition}[Acceptable timelines by contexts] \label{def:Acceptable timelines by contexts}
Let $\MM$ be a model and $\CC$ be a context for $\MM$. Define the set $\AT{\CC}$ of \defs{acceptable timelines by $\CC$} as $\EN{\CC}$.
\end{definition}

\noindent Note $\AT{\emptyset} = \TL{\MM}$. Acceptable timelines by $\CC$ are those meeting all the rules in $\CC$.

\begin{definition}[Contextualized models and contextualized pointed models] \label{def:Contextualized models and contextualized pointed models}
For every model $\MM$, context $\CC$, timeline $\pi$ in $\AT{\CC}$, and instant $i$, $(\MM, \CC)$ is called a \defs{contextualized model} and $(\MM, \CC, \pi, i)$ is called a \defs{contextualized pointed model}.
\end{definition}

\noindent For every contextualized pointed model $(\MM, \CC, \pi, i)$, the reason to require $\pi$ to be in $\AT{\CC}$ is that the timelines not in $\AT{\CC}$ are unacceptable for the agent and there is no point for her to consider them.

\begin{example}[Contextualized pointed models] \label{example:Contextualized pointed models}
Figure \ref{fig:Another contextualized pointed model} indicates a contextualized pointed model $(\MM, \CC, \pi_3, 1)$, where $\CC = \{\rul_1, \rul_2\}$, where $\rul_1 = \{\pi_1,\pi_2,\pi_3\}$ and $\rul_2 = \{\pi_2,\pi_3,\pi_4\}$. Then, $\AT{\CC} = \{\pi_2,\pi_3\}$, which means that $\pi_1$ and $\pi_4$ are unacceptable.

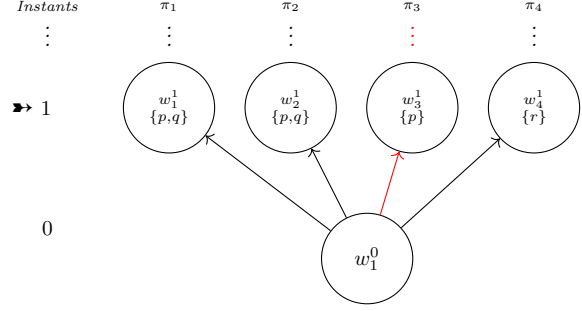
\begin{figure}[h]
\begin{center}
\begin{tikzpicture}[node distance=20mm,scale=0.8,every node/.style={transform shape}]
\tikzstyle{every state}=[draw=black,text=black,minimum size=15mm]

\node[state] (s-1-1) [] {$w^1_1 \atop \{p,q\}$};
\node[state] (s-1-2) [right of=s-1-1] {$w^1_2 \atop \{p,q\}$};
\node[state] (s-1-3) [right of=s-1-2] {$w^1_3 \atop \{p\}$};
\node[state] (s-1-4) [right of=s-1-3] {$w^1_4 \atop \{r\}$};
\node[state] (s-0-1) [below=25mm, right=5mm] at (s-1-2) {$w^0_1$};

\path
(s-0-1) edge [->,left] node {} (s-1-1)
(s-0-1) edge [->,left] node {} (s-1-2)
(s-0-1) edge [->,left,color=red] node {} (s-1-3)
(s-0-1) edge [->,left] node {} (s-1-4);

\node [above=9mm] at (s-1-1) {$\pi_1 \atop \vdots$};
\node [above=9mm] at (s-1-2) {$\pi_2 \atop \vdots$};
\node [above=9mm] at (s-1-3) {$\pi_3 \atop \red{\vdots}$};
\node [above=9mm] at (s-1-4) {$\pi_4 \atop \vdots$};

\node[] (1) [left of=s-1-1] {\hspace{-18pt} \ding{248} $1$};
\node[] (0) [below of=1] {$0$};
\node [above=9mm] at (1) {$\textit{Instants} \atop \vdots$};

\end{tikzpicture}
\end{center}

\caption{A contextualized pointed model}
\label{fig:Another contextualized pointed model}
\end{figure}
\end{example}

\subsection{Semantics}
\label{subsec:Semantics}

\begin{definition}[Semantics for $\Phi_{\ConSHNBT}$] \label{def:Semantics for Phi ConSONBT}

Fix a model $\MM$ and a context $\CC$ for $\MM$.

\bi
\item For every $\alpha \in \Phi_{\XY}$ and instant $i$, define the \defs{rule generated by $\alpha$ at $i$} as $\golf{\alpha}{i} = \{\pi \mid \MM, \CC, \pi, i \Vdash \alpha\}$.
\item For every $\alpha \in \Phi_{\XY}$ and instant $i$, define the \defs{update of $\CC$ with $\alpha$ at $i$} as $\upf{\CC}{\alpha}{i} = \upl{\CC}{\golf{\alpha}{i}}$.
\item \defs{Truth conditions of formulas of $\Phi_{\ConSHNBT}$} are defined as follows:

\medskip

\begin{tabular}{lll}
$\MM, \CC, \pi, i \Vdash p$ & $\Leftrightarrow$ & \parbox[t]{30em}{$\pi[i] \in V(p)$} \\
$\MM, \CC, \pi, i \not \Vdash \bot$ & & \\
$\MM, \CC, \pi, i \Vdash \neg \phi$ & $\Leftrightarrow$ & \parbox[t]{30em}{$\MM, \CC, \pi, i \not \Vdash \phi$} \\
$\MM, \CC, \pi, i \Vdash \phi \land \psi$ & $\Leftrightarrow$ & \parbox[t]{30em}{$\MM, \CC, \pi, i \Vdash \phi$ and $\MM, \CC, \pi, i \Vdash \psi$} \\
$\MM, \CC, \pi, i \Vdash \XXX \phi$ & $\Leftrightarrow$ & \parbox[t]{30em}{$\MM, \CC, \pi, i+1 \Vdash \phi$} \\
$\MM, \CC, \pi, i \Vdash \YYY \phi$ & $\Leftrightarrow$ & \parbox[t]{30em}{$\MM, \CC, \pi, i-1 \Vdash \phi$, if $0<i$} \\
$\MM, \CC, \pi, i \Vdash \con{\alpha} \phi$ & $\Leftrightarrow$ & \parbox[t]{30em}{$\MM, \upf{\CC}{\alpha}{i}, \pi', i \Vdash \phi$ for every $\pi' \in \AT{\upf{\CC}{\alpha}{i}}$}
\end{tabular}
\ei
\end{definition}

It can be verified that

\medskip

\begin{tabular}{lll}
$\MM, \CC, \pi, i \Vdash \dcon{\alpha} \phi$ & $\Leftrightarrow$ & \parbox[t]{27em}{$\MM, \upf{\CC}{\alpha}{i}, \pi', i \Vdash \phi$ for some $\pi' \in \AT{\upf{\CC}{\alpha}{i}}$} \\
$\MM, \CC, \pi, i \Vdash \Box \phi$ & $\Leftrightarrow$ & \parbox[t]{27em}{$\MM, \CC, \pi', i \Vdash \phi$ for every $\pi' \in \AT{\CC}$} \\
$\MM, \CC, \pi, i \Vdash \Diamond \phi$ & $\Leftrightarrow$ & \parbox[t]{27em}{$\MM, \CC, \pi', i \Vdash \phi$ for some $\pi' \in \AT{\CC}$}
\end{tabular}

We say a formula $\phi$ is \defs{valid} ($\models_\ConSHNBT \phi$) if $\MM, \CC, \pi, i \Vdash \phi$ for every contextualized pointed model $(\MM, \CC, \pi, i)$.

\begin{example} \label{example:tiger analysis}
We show how Example \ref{example:tiger} is analyzed in the formalization. We use $l, r$, and $a$ to respectively denote \emph{Adam is in the left room}, \emph{Adam is in the right room}, and \emph{Adam is alive}. Figure \ref{fig:pointed model tiger} indicates a model $\MM$. Let $\CC = \{\rul_1, \rul_2, \rul_3, \rul_4\}$ be a context, where $\rul_1 = \{\pi_1, \pi_3, \pi_4, \pi_5, \pi_6\}$, $\rul_2 = \{\pi_3, \pi_4, \pi_5, \pi_6\}$, $\rul_3 = \{\pi_1, \pi_2, \pi_3, \pi_5, \pi_6\}$ and $\rul_4 = \{\pi_1, \pi_2, \pi_3, \pi_4, \pi_5\}$. The ontic rules $\rul_1, \rul_2, \rul_3$ and $\rul_4$ respectively mean \emph{lions kill people}, \emph{people want to survive}\footnote{As Adam wants to survive, he will enter a room. This is why we use $\rul_2$ to indicate \emph{people want to survive}.}, \emph{tigers kill people}, and \emph{people die for a reason}.

The formula $\Box \XXX (l \lor r)$ says \emph{Adam must be in a room tomorrow}. It is true at $(\MM, \CC, \pi_5, 0)$. Why?

\bi
\item It can be verified $\AT{\CC} = \{\pi_3,\pi_5\}$.
\item $\MM, \CC, \pi', 0 \Vdash \XXX (l \lor r)$ for every $\pi' \in \AT{\CC}$. Thus, $\MM, \CC, \pi_5, 0 \Vdash \Box \XXX (l \lor r)$.
\ei

The formula $\con{\XXX l} \XXX \neg a$ says \emph{if Adam is in the left room tomorrow, then he must be dead}. It is true at $(\MM, \CC, \pi_5, 0)$. Why?

\bi
\item The update of $\CC$ with $\XXX l$ at $0$, $\upf{\CC}{\XXX l}{0}$, is $\{\golf{\XXX l}{0}, \rul_1, \rul_2, \rul_3, \rul_4\}$, where $\golf{\XXX l}{0} = \{\pi_3, \pi_4\}$. It can be verified $\AT{\upf{\CC}{\XXX l}{0}} = \{\pi_3\}$.
\item $\MM, \upf{\CC}{\XXX l}{0}, \pi', 0 \Vdash \XXX \neg a$ for every $\pi' \in \AT{\upf{\CC}{\XXX l}{0}}$. Thus, $\MM, \CC, \pi_5, 0 \Vdash \con{\XXX l} \XXX \neg a$.
\ei

The formula $\con{\XXX r} \XXX a$ says \emph{if Adam is in the right room tomorrow, then he must be alive}. It is true at $(\MM, \CC, \pi_5, 0)$. The argument is similar to the above one, and we skip it.

The formula $\Box (l \lor r)$ means \emph{Adam must be in a room now}. It is true at $(\MM, \CC, \pi_5, 1)$. Why?

\bi
\item As mentioned above, $\AT{\CC} = \{\pi_3,\pi_5\}$.
\item $\MM, \CC, \pi', 1 \Vdash l \lor r$ for every $\pi' \in \AT{\CC}$. Thus, $\MM, \CC, \pi_5, 1 \Vdash \Box (l \lor r)$.
\ei

The formula $\con{l} \neg a$ means \emph{if Adam were in the left room now, then he must be dead}. It is true at $(\MM, \CC, \pi_5, 1)$. Why?

\bi
\item The update of $\CC$ with $l$ at $1$, $\upf{\CC}{l}{1}$, is $\{\golf{l}{1}, \rul_1, \rul_2, \rul_3, \rul_4\}$, where $\golf{l}{1} = \{\pi_3, \pi_4\}$. It can be verified $\AT{\upf{\CC}{l}{1}} = \{\pi_3\}$.
\item $\MM, \upf{\CC}{l}{1}, \pi', 1 \Vdash \neg a$ for every $\pi' \in \AT{\upf{\CC}{l}{1}}$. Thus, $\MM, \CC, \pi_5, 1 \Vdash \con{l} \neg a$.
\ei

The formula $\con{r} a$ means \emph{since Adam is in the right room now, he must be alive}. It is true at $(\MM, \CC, \pi_5, 1)$. The argument is similar to the above one, and we skip it.

\begin{figure}[h]
\begin{center}
\begin{tikzpicture}[node distance=25mm,scale=0.8,every node/.style={transform shape}]
\tikzstyle{every state}=[draw=black,text=black,minimum size=14mm]

\node[state] (s-1-1) [] {$w^1_1 \atop \{\}$};
\node[state] (s-1-2) [right of=s-1-1] {$w^1_2 \atop \{a\}$};
\node[state] (s-1-3) [right of=s-1-2] {$w^1_3 \atop \{l\}$};
\node[state] (s-1-4) [right of=s-1-3] {$w^1_4 \atop \{l,a\}$};
\node[state] (s-1-5) [right of=s-1-4] {$w^1_5 \atop \{r,a\}$};
\node[state] (s-1-6) [right of=s-1-5] {$w^1_6 \atop \{r\}$};

\node[state] (s-0-1) [below=25mm, right=0mm] at (s-1-3) {$w^0_1 \atop \{\}$};

\node [above=12mm] at (s-1-1) {$\pi_1 \atop \vdots$};
\node [above=12mm] at (s-1-2) {$\pi_2 \atop \vdots$};
\node [above=12mm] at (s-1-3) {$\pi_3 \atop \vdots$};
\node [above=12mm] at (s-1-4) {$\pi_4 \atop \vdots$};
\node [above=12mm] at (s-1-5) {$\pi_5 \atop \vdots$};
\node [above=12mm] at (s-1-6) {$\pi_6 \atop \vdots$};

\path
(s-0-1) edge [->,left] node {} (s-1-1)
(s-0-1) edge [->,left] node {} (s-1-2)
(s-0-1) edge [->,left] node {} (s-1-3)
(s-0-1) edge [->,left] node {} (s-1-4)
(s-0-1) edge [->,left] node {} (s-1-5)
(s-0-1) edge [->,left] node {} (s-1-6);

\node[] (1) [left of=s-1-1] {$1$};
\node[] (0) [below=25mm] at (1) {$0$};
\node [above=12mm] at (1) {$\textit{Instants} \atop \vdots$};

\end{tikzpicture}
\end{center}

\caption{A pointed model for Example \ref{example:tiger}}
\label{fig:pointed model tiger}
\end{figure}
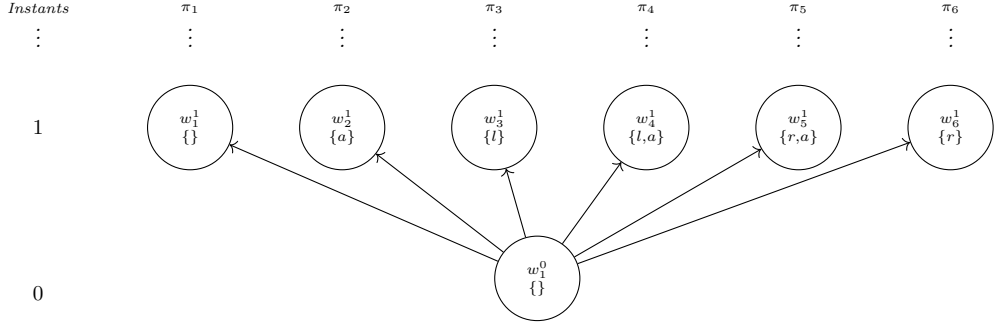
\end{example}

\subsection{Remarks}
\label{subsec:Remarks}


In the above, we define contexts as sets of ontic rules, which are sets of timelines. From the point of view of logic, contexts can be defined as sets of timelines without changing the set of valid formulas. We do this in order to make the formalization more intuitive.


The formulas $\alpha \rightarrow \Box \phi$, $\Box (\alpha \rightarrow \phi)$ and $\con{\alpha} \phi$ seem to have similar meaning, where $\alpha$ in $\Phi_{\XY}$ and $\phi$ in $\Phi_{\ConSHNBT}$. Now we compare their meaning.

\begin{itemize}
\item By the semantics, $\alpha \rightarrow \Box \phi$ is true at $(\MM, \CC, \pi, i)$ if and only if if $\alpha$ is true at $(\MM, \CC, \pi, i)$, then $\Box \phi$ is true at $(\MM, \CC, \pi, i)$. The formula $\alpha \rightarrow \Box \phi$ is a material implication: in every situation, it excludes a possibility about the truth values of $\alpha$ and $\Box \phi$ in this situation.
\item By the semantics, $\Box (\alpha \rightarrow \phi)$ is true at $(\MM, \CC, \pi, i)$ if and only if for every $\pi' \in \AT{\CC}$, if $\alpha$ is true at $(\MM, \CC, \pi', i)$, then $\phi$ is true at $(\MM, \CC, \pi', i)$. The formula $\Box (\alpha \rightarrow \phi)$ is the strict implication: in every situation, it, for any acceptable situation, excludes a possibility about the truth values of $\alpha$ and $\phi$ in the acceptable situation.
\item By the semantics, $\con{\alpha} \phi$ is true at $(\MM, \CC, \pi, i)$ if and only if for every $\pi' \in \AT{\upf{\CC}{\alpha}{i}}$, $\phi$ is true at $(\MM, \upf{\CC}{\alpha}{i}, \pi', i)$. In every situation, $\con{\alpha} \phi$ says that for every acceptable situation where $\alpha$ is true, $\phi$ holds with respect to the updated context.

The difference between $\alpha \rightarrow \Box \phi$ and $\con{\alpha} \phi$ can be shown by Fact \ref{fact:square brackets and implication}.

The difference between $\Box (\alpha \rightarrow \phi)$ and $\con{\alpha} \phi$ is that the former does not involve context change, but the latter does. This makes a difference for their truth values when $\phi$ contains a modality, as shown by Fact \ref{fact:square brackets and box}.
\end{itemize}

\begin{fact} \label{fact:square brackets and implication}
For all $\alpha$ in $\Phi_{\XY}$, $\con{\alpha} \alpha$ is valid, but for some $\alpha$ in $\Phi_{\XY}$, $\alpha \rightarrow \Box \alpha$ is invalid.
\end{fact}

\noindent This fact is easy to verify, and we skip its proof.

\begin{fact} \label{fact:square brackets and box}
For all $\alpha$ and $\beta$ in $\Phi_{\XY}$, $\con{\alpha} \beta \leftrightarrow \Box (\alpha \rightarrow \beta)$ is valid, but for some $\alpha$ in $\Phi_{\XY}$ and $\phi$ in $\Phi_{\ConSHNBT}$, $\con{\alpha} \phi \leftrightarrow \Box (\alpha \rightarrow \phi)$ is invalid.
\end{fact}

\begin{proof}
Let $\alpha$ and $\beta$ be in $\Phi_{\XY}$. Fix a contextualized pointed model $(\MM,\CC,\pi,i)$.
Assume $\MM,\CC,\pi,i \not \Vdash \con{\alpha} \beta$. Then, $\MM,\upf{\CC}{\alpha}{i},\pi',i \not \Vdash \beta$ for some $\pi' \in \AT{\upf{\CC}{\alpha}{i}}$. By Item \ref{item:basic alpha alpha} in Lemma \ref{lemma:basic}, $\pi' \in \golf{\alpha}{i}$ and $\pi' \in \AT{\CC}$. Then, $\MM,\CC,\pi',i \Vdash \alpha$ and $\MM,\CC,\pi',i \not \Vdash \beta$. Then, $\MM,\CC,\pi',i \not \Vdash \alpha \rightarrow \beta$. Then, $\MM,\CC,\pi,i \not \Vdash \Box (\alpha \rightarrow \beta)$.
Assume $\MM,\CC,\pi,i \not \Vdash \Box (\alpha \rightarrow \beta)$. Then, $\MM,\CC,\pi',i \not \Vdash \alpha \rightarrow \beta$ for some $\pi' \in \AT{\CC}$. Then, $\MM,\CC,\pi',i \Vdash \alpha$ and $\MM,\CC,\pi',i \not \Vdash \beta$. Then, $\pi' \in \golf{\alpha}{i}$. By Item \ref{item:basic alpha alpha} in Lemma \ref{lemma:basic}, $\pi' \in \AT{\upf{\CC}{\alpha}{i}}$. Then, $\MM,\upf{\CC}{\alpha}{i},\pi',i \not \Vdash \beta$. Then, $\MM,\CC,\pi,i \not \Vdash \con{\alpha} \beta$.

Let $\alpha = p$ and $\phi = \Box p$. Let $\MM = (W, <, V)$ be a model meeting the following conditions: it contains only two timelines, $\pi_1$ and $\pi_2$; $V(p) = \{\pi_2[1]\}$. Let $\CC$ be the empty context. It can be verified that $\con{p} \Box p$ is true at $(\MM, \pi_1, \CC, 1)$ but $\Box (p \rightarrow \Box p)$ is false there.
\end{proof}


The following fact indicates that unnested conditional strong historical necessity is monotonic, but nested conditional strong historical necessity is not.

\begin{fact}
For all $\alpha$, $\beta$ and $\gamma$ in $\Phi_{\XY}$, if $\alpha \rightarrow \beta$ is valid, then $\con{\beta} \gamma \rightarrow \con{\alpha} \gamma$ is valid, but for some $\alpha$ and $\beta$ in $\Phi_{\XY}$ and $\phi \in \Phi_{\ConSHNBT}$, $\alpha \rightarrow \beta$ is valid and $\con{\beta} \phi \rightarrow \con{\alpha} \phi$ is invalid.
\end{fact}

\begin{proof}
Let $\alpha$, $\beta$ and $\gamma$ be in $\Phi_{\XY}$. Assume $\alpha \rightarrow \beta$ is valid. Then, $\Box (\alpha \rightarrow \beta)$ is valid. Note that $\Box (\alpha \rightarrow \beta) \rightarrow (\Box (\beta \rightarrow \gamma) \rightarrow \Box (\alpha \rightarrow \gamma))$ is valid. Then, $\Box (\beta \rightarrow \gamma) \rightarrow \Box (\alpha \rightarrow \gamma)$ is valid. By Fact \ref{fact:square brackets and box}, $\con{\beta} \gamma \rightarrow \con{\alpha} \gamma$ is valid.

Let $\alpha = \top \land p$, $\beta = \top$ and $\phi = \Diamond \neg p$. Clearly, $(\top \land p) \rightarrow \top$ is valid. It is easy to verify that $\con{\top} \Diamond \neg p \rightarrow \con{\top \land p} \Diamond \neg p$ is not valid.
\end{proof}


The following fact indicates that reasoning by case is not sound with respect to conditional strong historical necessity.

\begin{fact}
For some $\alpha$ in $\Phi_{\XY}$, and $\phi$ and $\psi$ in $\Phi_{\ConSHNBT}$, the inference from $\con{\alpha} \phi$ and $\con{\neg \alpha} \psi$ to $\Box \phi \lor \Box \psi$ is unsound.
\end{fact}

\begin{proof}
Let $\alpha = p$, $\phi = p$ and $\psi = \neg p$. Let $\MM = (W, <, V)$ be a model meeting the following conditions: it contains only two timelines, $\pi_1$ and $\pi_2$; $V(p) = \{\pi_1[1]\}$. Let $\CC$ be the empty context. By Fact \ref{fact:square brackets and implication}, $\MM,\CC,\pi_1, 1 \Vdash \con{p} p$ and $\MM,\CC,\pi_1, 1 \Vdash \con{\neg p} \neg p$. It can be verified that $\MM,\CC,\pi_1, 1 \not \Vdash \Box p \lor \Box \neg p$.
\end{proof}

\subsection{Comparsions to some works on conditional futural necessity}
\label{subsec:Comparsions to some works on conditional futural necessity}

Here, we compare this work to \cite{ju_logic_2018} and \cite{thomason_theory_1980}, which explicitly and implicitly deal with conditional futural necessity, respectively.

Ju, Grilletti and Goranko \cite{ju_logic_2018} presented a logic $\mathsf{LTC}$ to formalize reasoning about temporal conditionals ``\exs{if $\phi$, it is necessary that $\psi$}''.
The logic $\mathsf{LTC}$ contains a next instant operator $\XXX$, a necessity operator $\mathbf{A}$, and an operator $[\cdot]$ of conditionals. The operators $\XXX$ and $\mathbf{A}$ are handled similarly to $\mathsf{CTL}^*$ proposed in \cite{emerson_sometimes_1986}. The operator $[\cdot]$ is dealt with as follows: $[\alpha] \phi$ is true at a pointed model $(\MM,\pi,i)$ if and only if $\phi$ is true at the pointed model $(\MM',\pi,i)$, where $\MM'$ is the result of shrinking $\MM$ by $\alpha$ at $(\pi,i)$.

The logic $\mathsf{LTC}$ and our work share a similar way of formalizing conditional necessity: the antecedent of conditional necessity updates something, and the consequent is evaluated relative to the update result. As a result, reasoning by case with respect to conditional necessity is sound in neither of them.

The main difference between the two theories is as follows. In the work of \cite{ju_logic_2018}, the antecedent of conditionals directly shrinks trees, while in our theory, it excludes timelines from sets of acceptable timelines. The consequences of the difference are mostly technical. Here, we mention only one point. Shrinking a tree and excluding timelines from a set of timelines of the tree are not the same, as some subset of the set might not correspond to a tree. Here, the so-called \emph{limit closure condition} for sets of timelines is involved \cite{reynolds_axiomatization_2001}.

Thomason and Gupta \cite{thomason_theory_1980} thought that time plays an essential role in evaluating conditionals and present two formal theories to make the role explicit. Both theories follow the approach of \cite{stalnaker_theory_1968} for conditionals: a conditional ``\exs{if $\phi$ then $\psi$}'' is true at a possibility $x$ if $\psi$ is true at \emph{the} possibility that resembles $x$ the most among the possibilities where $\phi$ is true. Both theories use branching time trees. The difference between the two theories is that in one of them, possibilities are pairs of a timeline and a moment, and in the other, they are pairs of a future choice function, which is a generation of timelines, and a moment. Thomason and Gupta introduce a special principle to compare possibilities, called \emph{the principle of past predominance}: if the past of a possibility $y$ is more similar to the past of a possibility $x$ than the past of a possibility $z$, then $y$ is more similar to $x$ than $z$.

The main difference between Thomason and Gupta's work and ours is that Thomason and Gupta use \emph{similarity} to deal with conditionals, while we use \emph{update of contexts} to address conditional strong historical necessity.

\subsection{Expressivity}
\label{subsec:Expressivity}

The operator $\con{\cdot}$ in $\Phi_{\ConSHNBT}$ can \emph{almost} be reduced.

\begin{lemma} \label{lemma:X and Y pass through everything}
The following formulas are valid:

\begin{enumerate}
\item
\begin{enumerate}
\item $\XXX \neg \phi \leftrightarrow \neg \XXX \phi$ \label{validity:X not}
\item $\XXX (\phi \land \psi) \leftrightarrow (\XXX \phi \land \XXX \psi)$ \label{validity:X and}
\item $\XXX \YYY \phi \leftrightarrow \phi$ \label{validity:X Y}
\item $\XXX \con{\alpha} \phi \leftrightarrow \con{\XXX \alpha} \XXX \phi$
\end{enumerate}
\item
\begin{enumerate}
\item $\YYY \neg \phi \leftrightarrow (\YYY \bot \lor \neg \YYY \phi)$ \label{validity:Y not}
\item $\YYY (\phi \land \psi) \leftrightarrow (\YYY \phi \land \YYY \psi)$ \label{validity:Y and}
\item $\YYY \XXX \phi \leftrightarrow (\YYY \bot \lor \phi)$ \label{validity:Y X}
\item $\YYY \con{\alpha} \phi \leftrightarrow \con{\YYY \alpha} \YYY \phi$
\end{enumerate}
\end{enumerate}
\end{lemma}

\noindent This result can be easily shown, and we skip its proof.

For every natural number $n$ (possibly $0$), we use $\XXX^n$ and $\YYY^n$ to indicate the sequences of $n$ $\XXX$s and $n$ $\YYY$s respectively.

\begin{definition}[The languages $\Phi_{\NXXYY}$ and $\Phi_{\ConXXYY}$] \label{def:The Languages Phi NXXYY and ConXXYY}
Define the language $\Phi_{\NXXYY}$ as follows:

\[\beta ::= \XXX^n p \mid \XXX^n \bot \mid \YYY^n p \mid \YYY^n \bot \mid \neg \beta \mid (\beta \land \beta)\]

\noindent Define the language $\Phi_{\ConXXYY}$ as follows, where $\beta \in \Phi_{\NXXYY}$:

\[\phi ::= \XXX^n p \mid \XXX^n \bot \mid \YYY^n p \mid \YYY^n \bot \mid \neg \phi \mid (\phi \land \phi) \mid \con{\beta} \phi\]
\end{definition}

By Lemma \ref{lemma:X and Y pass through everything}, we can easily show the following result, whose proof is skipped.

\begin{fact} \label{fact:reduction 1 semantics}
There is an effective function $\kappa$ from $\Phi_{\ConSHNBT}$ to $\Phi_{\ConXXYY}$ such that for every $\phi \in \Phi_{\ConSHNBT}$, $\models_\ConSHNBT \phi \leftrightarrow \kappa(\phi)$.
\end{fact}

\begin{definition}[Closed formulas] \label{def:Closed formulas}
Define closed formulas of $\Phi_{\ConSHNBT}$ as follows, where $\alpha \in \Phi_{\XY}$ and $\phi \in \Phi_{\ConSHNBT}$:

\[\chi ::= \con{\alpha} \phi \mid \neg \chi \mid (\chi \land \chi)\]
\end{definition}

The following fact indicates that the truth values of closed formulas at contextualized pointed models $(\MM,\CC,\pi,i)$ are independent of $\pi$:

\begin{fact} \label{fact:closed formulas}
Let $\chi$ be a closed formula. Let $\MM$ be a model, $\CC$ be a context, and $i$ be an instant. Then, $\MM,\CC,\pi,i \Vdash \chi$ if and only if $\MM,\CC,\pi',i \Vdash \chi$ for all $\pi$ and $\pi'$ in $\AT{\CC}$.
\end{fact}

\begin{lemma} \label{lemma:partial reduction conditionals valid}
The following formulas are valid, where $\alpha, \beta$ and $\gamma$ are in $\Phi_{\XY}$:
\begin{enumerate}
\item $\con{\alpha} (\phi \land \psi) \leftrightarrow (\con{\alpha} \phi \land \con{\alpha} \psi)$ \label{validity:conditionals distribution conjunction}
\item $\con{\alpha} (\phi \lor \chi) \leftrightarrow (\con{\alpha} \phi \lor \con{\alpha} \chi)$, where $\chi$ is a closed formula \label{validity:conditionals distribution conditionally distribution}
\item $\con{\alpha} \con{\beta} \gamma \leftrightarrow \con{\alpha \land \beta} \gamma$ \label{validity:conditionals conditionals}
\item $\con{\alpha} \dcon{\beta} \gamma \leftrightarrow (\con{\alpha} \bot \lor \dcon{\alpha \land \beta} \gamma)$ \label{validity:conditionals dual conditionals}
\item $\con{\alpha} \beta \leftrightarrow \Box (\alpha \rightarrow \beta)$ \label{validity:conditionals box}
\end{enumerate}
\end{lemma}

\noindent The proof for this result is put in Section \ref{sec:Proofs for expressivity} in the Appendix.

\begin{definition}[The language $\Phi_{\OneBXXYY}$] \label{def:language Phi OneBXXYY}
The language $\Phi_{\OneBXXYY}$ is defined as follows, where $\beta \in \Phi_{\NXXYY}$:
\[\phi ::= \beta \mid \neg \phi \mid (\phi \land \phi) \mid \Box \beta\]
\end{definition}

\begin{fact} \label{fact:reduction 2 semantics}
There is an effective function $\mu$ from $\Phi_{\ConXXYY}$ to $\Phi_{\OneBXXYY}$ such that for every $\phi \in \Phi_{\ConXXYY}$, $\models_\ConSHNBT \phi \leftrightarrow \mu(\phi)$.
\end{fact}

\noindent The proof for this result is put in Section \ref{sec:Proofs for expressivity} in the Appendix.

By Facts \ref{fact:reduction 1 semantics} and \ref{fact:reduction 2 semantics}, we know $\Phi_{\ConSHNBT}$ is as expressive as $\Phi_{\OneBXXYY}$.

\subsection{Axiomatization}
\label{subsec:Axiomatization}

We use $\Phi_{\PL}$ to denote the language of the Propositional Logic.

\begin{definition}[Axiomatic system $\ConSHNBT$] \label{def:Axiomatic system ConSONBT}

Define an axiomatic system $\ConSHNBT$ as follows:

\medskip

\noindent Axioms:

\begin{enumerate}
\item Axioms for the Propositional Logic
\item Axioms for $\XXX$:
\begin{enumerate}
\item $\XXX \neg \phi \leftrightarrow \neg \XXX \phi$ \label{E axiom:X not}
\item $\XXX (\phi \land \psi) \leftrightarrow (\XXX \phi \land \XXX \psi)$ \label{E axiom:X and}
\item $\XXX \YYY \phi \leftrightarrow \phi$ \label{E axiom:X Y}
\item $\XXX \con{\alpha} \phi \leftrightarrow \con{\XXX \alpha} \XXX \phi$ \label{E axiom:X con}
\item $\neg \XXX \neg \top$ \label{E axiom:X dual top}
\end{enumerate}
\item Axioms for $\YYY$:
\begin{enumerate}
\item $\YYY \neg \phi \leftrightarrow (\YYY \bot \lor \neg \YYY \phi)$ \label{E axiom:Y not}
\item $\YYY (\phi \land \psi) \leftrightarrow (\YYY \phi \land \YYY \psi)$ \label{E axiom:Y and}
\item $\YYY \XXX \phi \leftrightarrow (\YYY \bot \lor \phi)$ \label{E axiom:Y X}
\item $\YYY \con{\alpha} \phi \leftrightarrow \con{\YYY \alpha} \YYY \phi$ \label{E axiom:Y con}
\item $\Diamond \YYY \bot \rightarrow (\Diamond \alpha \rightarrow \alpha)$, where $\alpha \in \Phi_{\PL}$ \label{E axiom:if Y bot}
\end{enumerate}
\item Axioms for partial reduction of $\con{\cdot}$, where $\alpha, \beta$ and $\gamma$ are in $\Phi_{\XY}$:
\begin{enumerate}
\item $\con{\alpha} (\phi \land \psi) \leftrightarrow (\con{\alpha} \phi \land \con{\alpha} \psi)$ \label{E axiom:con and}
\item $\con{\alpha} (\phi \lor \chi) \leftrightarrow (\con{\alpha} \phi \lor \con{\alpha} \chi)$, where $\chi$ is a closed formula \label{E axiom:con or}
\item $\con{\alpha} \con{\beta} \gamma \leftrightarrow \con{\alpha \land \beta} \gamma$ \label{E axiom:con con}
\item $\con{\alpha} \dcon{\beta} \gamma \leftrightarrow \dcon{\alpha \land \beta} \gamma$ \label{E axiom:con dual con}
\item $\con{\alpha} \beta \leftrightarrow \Box (\alpha \rightarrow \beta)$ \label{axiom:conditionals box}
\end{enumerate}
\item Axioms for $\Box$: $\Box \alpha \rightarrow \alpha$, where $\alpha \in \Phi_{\XY}$ \label{E axiom:A T}
\end{enumerate}

\noindent Inference rules:

\begin{enumerate}
\item Modus ponens: from $\phi$ and $\phi \rightarrow \psi$, we can get $\psi$;
\item Generalization of $\XXX$: from $\phi$, we can get $\XXX \phi$;
\item Generalization of $\YYY$: from $\phi$, we can get $\YYY \phi$;
\item Generalization of $\Box$: from $\alpha$, we can get $\Box \alpha$, where $\alpha \in \Phi_{\XY}$;
\item Replacement of equivalent subformulas: from $\psi \leftrightarrow \psi'$, we can get $\phi \leftrightarrow \phi'$, where $\psi$ is a subformula of $\phi$ and $\phi'$ is the result of replacing $\psi$ by $\psi'$ in $\phi$.
\end{enumerate}
\end{definition}

\noindent We use $\vdash_{\ConSHNBT} \phi$ to indicate $\phi$ is \defs{derivable} in this system.

\begin{theorem} \label{theorem:soundness completeness}
The axiomatic system $\ConSHNBT$ is sound and complete with respect to the set of valid formulas of $\Phi_{\ConSHNBT}$.
\end{theorem}

\noindent We put the proof for this result in Section \ref{sec:Proofs for axiomatization} in the Appendix.

\section{Analyses of the nontheological version of Lavenham's argument}
\label{sec:Analyses of the nontheological version of Lavenham's argument}

\subsection{The argument and our analyses of it}
\label{subsec:The argument and our analyses of it}

The nontheological version of Lavenham's argument from \cite{ohrstrom_future_2020} goes as follows, where E indicates an \emph{event} and non-E is its \emph{negation}.

\begin{enumerate}
\item Either E is going to take place tomorrow or non-E is going to take place tomorrow. (Assumption).
\item If a proposition about the past is true, then it is now necessary. (Assumption).
\item If E is going to take place tomorrow, then yesterday E would take place in two days. (Assumption).
\item If E is going to take place tomorrow, then it is now necessary that yesterday E would take place in two days. (Follows from 2 and 3).
\item If it is now necessary that yesterday E would take place in two days, then it is now necessary that E is going to take place tomorrow. (Assumption).
\item If E is going to take place tomorrow, then E is necessarily going to take place tomorrow. (Follows from 4 and 5).
\item If non-E is going to take place tomorrow, then non-E is necessarily going to take place tomorrow. (Follows by the same kind of reasoning as 6).
\item Either E is necessarily going to take place tomorrow or non-E is necessarily going to take place tomorrow. (Follows from 1, 6 and 7).
\item Therefore, what is going to happen tomorrow is going to happen with necessity. (Follows from 8).
\end{enumerate}

We treat \emph{E/non-E taking place} as \emph{$p$/$\neg p$ being true}, where $p$ is a proposition. We notice that \emph{events taking place} and \emph{propositions being true} might not be the same, but we do not go into this issue in this work. We suppose the notion of necessity in this argument is strong historical necessity. This is an issue deserving some philosophical consideration. Again, we do not go into it here.

We respectively consider how the nine statements in the argument are expressed in our logic.

\begin{enumerate}
\item We express Statement 1 as $\XXX p \lor \XXX \neg p$. In the literature, this statement is called \emph{the principle of excluded future middle}.
\item Statement 2 is a schema, and we cannot express it. However, only two of its instantiations are used in the argument: (a) \emph{if it is true that yesterday E would take place in two days, then it is now necessary that yesterday E would take place in two days;} (b) \emph{if it is true that yesterday non-E would take place in two days, then it is now necessary that yesterday non-E would take place in two days.} The two statements are conditional strong historical necessity, and we express them as $\con{\YYY \XXX \XXX p} \YYY \XXX \XXX p$ and $\con{\YYY \XXX \XXX \neg p} \YYY \XXX \XXX \neg p$, respectively. Statement 2 is called \emph{the principle of the necessity of the past} in the literature.
\item Statement 3 says that \emph{E is going to take place tomorrow} logically implies \emph{yesterday E would take place in two days}. It is a material implication, and we express it as $\XXX p \rightarrow \YYY \XXX \XXX p$.
\item Statement 4 is conditional strong historical necessity, and we express it as $\con{\XXX p} \YYY \XXX \XXX p$.
\item Statement 5 says that \exs{it is now necessary that yesterday E would take place in two days} logically implies \emph{it is now necessary that E is going to take place tomorrow}. It is a material implication, and we express it as $\Box \YYY \XXX \XXX p \rightarrow \Box \XXX p$.
\item Statement 6 is conditional strong historical necessity, and we express it as $\con{\XXX p} \XXX p$.
\item Statement 7 is conditional strong historical necessity, and we express it as $\con{\XXX \neg p} \XXX \neg p$.
\item Statement 8 is expressed as $\Box \XXX p \lor \Box \XXX \neg p$.
\item Statement 9 is a schema we cannot express.
\end{enumerate}

Suppose we just consider the two instantiations of Statement 2 and do not consider Statement 9. Then, the argument can be formalized in our logic, as Figure \ref{fig:Representation of the core part of the argument in ConSONBT} shows.

\begin{figure}[h]
\centering
\begin{forest}
for tree=
{
rounded corners,
text centered,
grow=90,
forked edge,
l sep=8mm,
fork sep=4mm,
s sep=4mm
}
[{8: $\Box \XXX p \lor \Box \XXX \neg p$},draw
[{7: $\con{\XXX \neg p} \XXX \neg p$},draw
[{$\vdots$}]
]
[{6: $\con{\XXX p} \XXX p$},draw
[{5: $\Box \YYY \XXX \XXX p \rightarrow \Box \XXX p$},fill=gray!10,draw]
[{4: $\con{\XXX p} \YYY \XXX \XXX p$},draw
[{3: $\XXX p \rightarrow \YYY \XXX \XXX p$},fill=gray!10,draw]
[{2: $\con{\YYY \XXX \XXX p} \YYY \XXX \XXX p$},fill=gray!10,draw]
]
]
[{1: $\XXX p \lor \XXX \neg p$},fill=gray!10,draw]
]
\end{forest}

\caption{Formalization of the nontheological version of Lavenham's argument in $\ConSHNBT$.}
\label{fig:Representation of the core part of the argument in ConSONBT}
\end{figure}
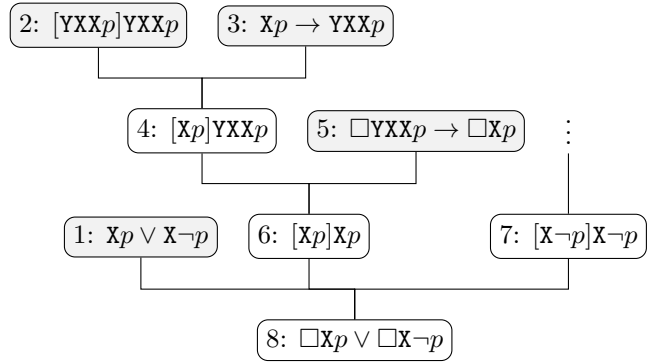

It can be verified that all premises of the argument are valid but its conclusion is not. Therefore, the argument is not sound.

\subsection{Comparisons to some related works}
\label{subsec:Comparisons to some related works}

Scholars have developed various formal theories rejecting similar arguments for futural determinism. These theories, among others, include the Peircean temporal logic by Prior \cite{prior_past_1967}, the Ockhamist temporal logic by Prior \cite{prior_past_1967}, the supervaluationist theory by Thomason \cite{thomason_indeterminist_1970}, the true futurist theory by {\O}hrstr{\o}m \cite{ohrstrom_problems_1981}, the relativist theory by MacFarlane \cite{macfarlane_future_2003}, and the logic $\mathsf{LTC}$ for temporal conditionals by Ju, Grilletti and Goranko \cite{ju_logic_2018}, which contains a solution to Aristotle's sea battle puzzle.

All of these theories use branching time models, as we do.

All these theories treat the necessity in the arguments they consider as the futural necessity. This can be seen from the fact that $p \rightarrow \Box p$, where $p$ is an atomic proposition, is valid in all of them. We treat the necessity in the nontheological version of Lavenham's argument as strong historical necessity, and $p \rightarrow \Box p$ is not valid in our logic.

All of these theories except $\mathsf{LTC}$ formalize the principle of the necessity of the past as a material implication, which is not generally valid in them. This principle is expressed as $\con{\YYY \phi} \YYY \phi$ in our logic, which is valid.

All these theories except $\mathsf{LTC}$ accept the soundness of the arguments they consider and argue that their premises are not all valid. By our logic, all premises of the nontheological version of Lavenham's argument are valid, but the argument is not sound.

Ju, Grilletti and Goranko \cite{ju_logic_2018} focused on the following version of Aristotle's sea battle puzzle, which is from \cite{garrett_what_2017}\footnote{The texts in \emph{On Interpretation} concerning Aristotle's argument for future determinism are unclear, and the argument has different versions in the literature. The texts of \emph{On Interpretation} show that Aristotle's argument also involves the past. So, this argument simplifies things.}. \emph{Either there will be a sea battle tomorrow or not. If there is a sea battle tomorrow, then it is necessarily so. If there is no sea battle tomorrow, then it is necessarily so. Thus, either there will necessarily be a sea battle tomorrow or there will necessarily be no sea battle tomorrow.} This argument is formalized in $\mathsf{LTC}$ as follows, where $s$ indicate \emph{there is a sea battle}: $\XXX s \lor \XXX \neg s,$ $[\XXX s] \mathbf{A} \XXX s,$ $[\XXX \neg s] \mathbf{A} \XXX \neg s$ $\models$ $\mathbf{A} \XXX s \lor \mathbf{A} \XXX \neg s$. The three premises are valid, but the conclusion is not in $\mathsf{LTC}$. Therefore, the argument is not sound. This solution and ours are similar.

\section{Conclusion}
\label{sec:Conclusion}

In this work, we presented a logic for conditional strong historical necessity in branching time and applied it to analyze the nontheological version of Lavenham's argument for future determinism. The approach for the logic is as follows. The agent accepts a set of indefeasible ontic rules concerning how the world evolves over time, called a context. A context determines a set of acceptable timelines, which is the domain for strong historical necessity. When evaluating a sentence with conditional strong historical necessity with respect to a context, we put the rule corresponding to its antecedent to the context and then check whether its consequent holds for all acceptable timelines by the new context with respect to the new context. By our logic, the argument is unsound.

The logic we propose is not very expressive. It contains only two simple temporal operators: the next instant operator and the last instant operator. It will be natural to introduce the operators \emph{until} and \emph{since} in the future.

\subsection*{Acknowledgments}

I want to thank Valentin Goranko for his great help with this project. Thanks also go to the audience of logic seminars at Beijing Normal University and Hubei University. This research was supported by the National Social Science Foundation of China (No. 19BZX137).

\bibliographystyle{alpha}
\bibliography{Modalities,Special-items}

\appendix

\newcommand{\ecai}{\AT{\upf{\CC}{\alpha}{i}}}

\newcommand{\mcpiit}{\MM, \CC, \pi, i \Vdash}
\newcommand{\mcpiif}{\MM, \CC, \pi, i \not \Vdash}
\newcommand{\mcpipit}{\MM, \CC, \pi', i \Vdash}
\newcommand{\mcpipif}{\MM, \CC, \pi', i \not \Vdash}
\newcommand{\mcapipit}{\MM, \upf{\CC}{\alpha}{i}, \pi', i \Vdash}
\newcommand{\mcapipif}{\MM, \upf{\CC}{\alpha}{i}, \pi', i \not \Vdash}
\newcommand{\mcapippit}{\MM, \upf{\CC}{\alpha}{i}, \pi'', i \Vdash}
\newcommand{\mcapippif}{\MM, \upf{\CC}{\alpha}{i}, \pi'', i \not \Vdash}

\section{Proofs for expressivity}
\label{sec:Proofs for expressivity}

\begin{lemma} \label{lemma:basic}
Fix a model.

\begin{enumerate}
\item $\AT{\upf{(\upf{\CC}{\alpha}{i})}{\beta}{i}} = \AT{\upf{\CC}{(\alpha \land \beta)}{i}}$ \label{item:basic alpha plus beta alpha and beta}
\item $\AT{\upf{\CC}{\alpha}{i}} = \AT{\CC} \cap \golf{\alpha}{i}$ \label{item:basic alpha alpha}
\end{enumerate}
\end{lemma}

\noindent This result is easy to show.

\setcounter{rlemma}{1}

\begin{rlemma}
The following formulas are valid, where $\alpha, \beta$ and $\gamma$ are in $\Phi_{\XY}$:
\begin{enumerate}
\item $\con{\alpha} (\phi \land \psi) \leftrightarrow (\con{\alpha} \phi \land \con{\alpha} \psi)$ \label{r validity:conditionals distribution conjunction}
\item $\con{\alpha} (\phi \lor \chi) \leftrightarrow (\con{\alpha} \phi \lor \con{\alpha} \chi)$, where $\chi$ is a closed formula \label{r validity:conditionals distribution conditionally distribution}
\item $\con{\alpha} \con{\beta} \gamma \leftrightarrow \con{\alpha \land \beta} \gamma$ \label{r validity:conditionals conditionals}
\item $\con{\alpha} \dcon{\beta} \gamma \leftrightarrow (\con{\alpha} \bot \lor \dcon{\alpha \land \beta} \gamma)$ \label{r validity:conditionals dual conditionals}
\item $\con{\alpha} \beta \leftrightarrow \Box (\alpha \rightarrow \beta)$ \label{r validity:conditionals box}
\end{enumerate}
\end{rlemma}

\begin{proof}
The first two items are easy, the last one is already shown before, and we just show the third and fourth ones.

3. Fix a contextualized pointed model $(\MM,\CC,\pi,i)$. Note $\AT{\upf{(\upf{\CC}{\alpha}{i})}{\beta}{i}} = \AT{\upf{\CC}{(\alpha \land \beta)}{i}}$ by Item \ref{item:basic alpha plus beta alpha and beta} in Lemma \ref{lemma:basic}.

Assume $\mcpiif \con{\alpha}\con{\beta} \gamma$. Then, $\mcapipif \con{\beta} \gamma$ for some $\pi' \in \ecai$. Then, $\MM,\upf{(\upf{\CC}{\alpha}{i})}{\beta}{i},\pi'',i \not \Vdash \gamma$ for some $\pi'' \in \AT{\upf{(\upf{\CC}{\alpha}{i})}{\beta}{i}}$. Then, $\pi'' \in \AT{\upf{\CC}{(\alpha \land \beta)}{i}}$. Then, $\MM,\upf{\CC}{(\alpha \land \beta)}{i},\pi'',i \not \Vdash \gamma$. Then, $\mcpiif \con{\alpha \land \beta} \gamma$.

Assume $\mcpiif \con{\alpha \land \beta} \gamma$. Then, $\MM,\upf{\CC}{(\alpha \land \beta)}{i},\pi',i \not \Vdash \gamma$ for some $\pi' \in \AT{\upf{\CC}{(\alpha \land \beta)}{i}}$. Note $\pi' \in \AT{\upf{(\upf{\CC}{\alpha}{i})}{\beta}{i}}$. Then, $\MM,\upf{(\upf{\CC}{\alpha}{i})}{\beta}{i},\pi',i \not \Vdash \gamma$. By by Item \ref{item:basic alpha plus beta alpha and beta} in Lemma \ref{lemma:basic}, $\pi' \in \AT{\upf{\CC}{\alpha}{i}}$. Then, $\mcapipif \con{\beta} \gamma$. Then, $\mcpiif \con{\alpha}\con{\beta} \gamma$.

4. Fix a contextualized pointed model $(\MM,\CC,\pi,i)$. Note $\AT{\upf{(\upf{\CC}{\alpha}{i})}{\beta}{i}} = \AT{\upf{\CC}{(\alpha \land \beta)}{i}}$ by Item \ref{item:basic alpha plus beta alpha and beta} in Lemma \ref{lemma:basic}.

Assume $\AT{\upf{\CC}{\alpha}{i}} = \emptyset$. Then, both sides of the equivalence hold at $(\MM,\CC,\pi,i)$ trivially.

Assume $\AT{\upf{\CC}{\alpha}{i}} \neq \emptyset$.

Assume $\mcpiit \con{\alpha}\dcon{\beta} \gamma$. Let $\pi' \in \ecai$. Then, $\mcapipit \dcon{\beta}\gamma$. Then, there is $\pi'' \in \AT{\upf{(\upf{\CC}{\alpha}{i})}{\beta}{i}}$ such that $\MM,\upf{(\upf{\CC}{\alpha}{i})}{\beta}{i},\pi'',i \Vdash \gamma$. Then, $\pi'' \in \AT{\upf{\CC}{(\alpha \land \beta)}{i}}$ and $\MM,\upf{\CC}{(\alpha \land \beta)}{i},\pi'',i \Vdash \gamma$. Then, $\mcpiit \dcon{\alpha \land \beta} \gamma$. Then, $\mcpiit \con{\alpha} \bot \lor \dcon{\alpha \land \beta} \gamma$.

Assume $\mcpiit \con{\alpha} \bot \lor \dcon{\alpha \land \beta} \gamma$. Then, $\MM,\CC,\pi,i \Vdash \dcon{\alpha \land \beta} \gamma$. Then, $\MM,\upf{\CC}{(\alpha \land \beta)}{i},\pi',i \Vdash \gamma$ for some $\pi' \in \AT{\upf{\CC}{(\alpha \land \beta)}{i}}$. Then, $\pi' \in \AT{\upf{(\upf{\CC}{\alpha}{i})}{\beta}{i}}$ and $\MM,\upf{(\upf{\CC}{\alpha}{i})}{\beta}{i},\pi',i \Vdash \gamma$. Let $\pi'' \in \ecai$. Then, $\MM,\upf{\CC}{\alpha}{i},\pi'',i \Vdash \dcon{\beta}\gamma$. Then, $\mcpiit \con{\alpha}\dcon{\beta} \gamma$.
\end{proof}

\setcounter{rfact}{4}

\begin{rfact}
There is an effective function $\mu$ from $\Phi_{\ConXXYY}$ to $\Phi_{\OneBXXYY}$ such that for every $\phi \in \Phi_{\ConXXYY}$, $\models_\ConSHNBT \phi \leftrightarrow \mu(\phi)$.
\end{rfact}

\begin{proof}
~

Define the modal depth of formulas of $\Phi_{\ConSHNBT}$ with respect to $\con{\cdot}$ in the usual way. Fix a formula $\phi$ in $\Phi_{\ConXXYY}$.

Repeat the following steps until we cannot proceed.

\bi
\item Pick a sub-formula $\con{\alpha} \psi$ of $\phi$ whose modal depth with respect to $\con{\cdot}$ is $2$, if $\phi$ has such a sub-formula.
\item Transform $\psi$ to $\chi_1 \land \dots \land \chi_n$, where all $\chi_i$ is in the form of $(\beta_1 \lor \dots \lor \beta_k) \lor (\con{\gamma_1} \zeta_1 \lor \dots \lor \con{\gamma_l} \zeta_l) \lor (\dcon{\eta_1} \theta_1 \lor \dots \lor \dcon{\eta_m} \theta_m)$, where all $\beta_i, \gamma_i, \zeta_i, \eta_i$ and $\theta_i$ are in $\Phi_{\NXXYY}$.
\item Note $\models_{\ConSHNBT} \con{\alpha} \psi \leftrightarrow (\con{\alpha}\chi_1 \land \dots \land \con{\alpha}\chi_n)$ by Item \ref{validity:conditionals distribution conjunction} in Lemma \ref{lemma:partial reduction conditionals valid}. Repeat the following steps until we cannot proceed:
\bi
\item From $\con{\alpha}\chi_1 \land \dots \land \con{\alpha}\chi_n$, pick a conjunct $\con{\alpha} \chi_i = \con{\alpha} ((\beta_1 \lor \dots \lor \beta_k) \lor (\con{\gamma_1} \zeta_1 \lor \dots \lor \con{\gamma_l} \zeta_l) \lor (\dcon{\eta_1} \theta_1 \lor \dots \lor \dcon{\eta_m} \theta_m))$.
\item By Item \ref{validity:conditionals distribution conditionally distribution} in Lemma \ref{lemma:partial reduction conditionals valid}, $\con{\alpha} \chi_i \leftrightarrow \xi$ is valid, where $\xi = \con{\alpha}(\beta_1 \lor \dots \lor \beta_k) \lor (\con{\alpha}\con{\gamma_1} \zeta_1 \lor \dots \lor \con{\alpha}\con{\gamma_l} \zeta_l) \lor (\con{\alpha}\dcon{\eta_1} \theta_1 \lor \dots \lor \con{\alpha}\dcon{\eta_m} \theta_m)$. In the ways specified by Items \ref{validity:conditionals conditionals} and \ref{validity:conditionals dual conditionals} in Lemma \ref{lemma:partial reduction conditionals valid}, transform $\xi$ to $\xi'$, whose modal depth with respect to $\con{\cdot}$ is $1$.
\item Replace $\con{\alpha} \chi_i$ by $\xi'$ in $\con{\alpha}\chi_1 \land \dots \land \con{\alpha}\chi_n$.
\ei
\ei

\noindent Let $\phi'$ be the result. It is easy to see that $\phi'$ contains no nested conditionals and $\models_{\ConSHNBT} \phi \leftrightarrow \phi'$.

By Item \ref{validity:conditionals box} in Lemme \ref{lemma:partial reduction conditionals valid}, we can get a formula $\phi''$ in $\Phi_{\OneConXXYY}$ such that $\models_{\ConSHNBT} \phi' \leftrightarrow \phi''$.

Define $\mu(\phi)$ as $\phi''$. Then, $\models_{\ConSHNBT} \phi \leftrightarrow \mu(\phi)$. We can see that the procedure of getting $\mu(\phi)$ is effective.
\end{proof}

\section{Proofs for axiomatization}
\label{sec:Proofs for axiomatization}

\newcommand{\HX}{\mathsf{H}}
\newcommand{\IX}{\mathsf{I}}
\newcommand{\LX}{\mathsf{L}}

\newcommand{\DJX}[1]{\mathbf{DJ}(#1)}

\newcommand{\CF}{\textrm{CF}}

\newcommand{\fcf}{\Box \HX \land \Diamond \IX_1 \land \dots \land \Diamond \IX_i \land \LX}

\newcommand{\BS}{\textrm{BS}}

\newcommand{\fbs}{(\HX \land \IX_1, \dots, \HX \land \IX_i, \HX \land \LX)}

\newcommand{\AS}{\textrm{AS}}

\newcommand{\fas}{(\HX'_1 \land \IX'_1, \dots, \HX'_i \land \IX'_i, \HX'' \land \LX')}

\newcommand{\NN}{\mathtt{N}}

\subsection{Soundness and completeness of an axiomatic system $\OneBXXYY$ in the language $\Phi_{\OneBXXYY}$}
\label{subsec:Soundness and completeness of an axiomatic system OneBXXYY in the language Phi OneBXXYY}

\begin{definition}[Axiomatic system $\OneBXXYY$] \label{def:Axiomatic system OneBXXYY}

Define an axiomatic system $\OneBXXYY$ as follows:

\medskip

\noindent Axioms:

\begin{enumerate}
\item Axioms for the Propositional Logic
\item Axioms for $\XXX$: $\neg \XXX \neg \top$ \label{AXXYY axiom:X dual top}
\item Axioms for $\YYY$: $\Diamond \YYY \bot \rightarrow (\Diamond \alpha \rightarrow \alpha)$, where $\alpha \in \Phi_{\PL}$ \label{AXXYY axiom:if Y bot}
\item Axioms for $\Box$:
\begin{enumerate}
\item $\Box (\alpha \rightarrow \beta) \rightarrow (\Box \alpha \rightarrow \Box \beta)$, where $\alpha$ and $\beta$ are in $\Phi_{\NXXYY}$ \label{AXXYY axiom:A K}
\item $\Box \alpha \rightarrow \alpha$, where $\alpha \in \Phi_{\NXXYY}$ \label{AXXYY axiom:A T}
\end{enumerate}
\end{enumerate}

\noindent Inference rules:

\begin{enumerate}
\item Modus ponens: from $\phi$ and $\phi \rightarrow \psi$, we can get $\psi$;
\item Generalization of $\XXX$: from $\phi$, we can get $\XXX \phi$;
\item Generalization of $\YYY$: from $\phi$, we can get $\YYY \phi$;
\item Generalization of $\Box$: from $\alpha$, we can get $\Box \alpha$, where $\alpha \in \Phi_{\XY}$.
\end{enumerate}
\end{definition}

\noindent We use $\vdash_{\OneBXXYY} \phi$ to indicate that $\phi$ is \defs{derivable} in this system.

We call $p, \neg p, \bot$ and $\top$ literals. Some auxiliary notations follow:

\bi

\item

For every $\phi \in \Phi_{\NXXYY}$, we use $\DJX{\phi}$ to indicate the set of disjuncts of the disjunctive normal form of $\phi$. Note the elements of $\DJX{\phi}$ are in the form of $\XXX^{h_1} l_1 \land \dots \land \XXX^{h_n} l_n \land \YYY^{j_1} l'_1 \land \dots \land \YYY^{j_m} l'_m$, where all $l_x$ and $l'_x$ are literals.

\item

A formula $\CF$ is called a \emph{core formula} if it is in the form of $\fcf$, where $\HX, \IX_1, \dots, \IX_i$, and $\LX$ are all in $\Phi_{\NXXYY}$.

\item

Fix a core formula $\CF = \fcf$. The following sequence $\BS$ is called the \emph{basic sequence} for $\CF$: $\fbs$. The following sequence $\AS$ is called an \emph{atomic sequence} for $\CF$: $\fas$, where $\HX'_1, \dots, \HX'_i, \HX'' \in \DJX{\HX}$, $\IX'_1 \in \DJX{\IX_1}$, \dots, $\IX'_i \in \DJX{\IX_i}$, and $\LX' \in \DJX{\LX}$.

\item Let $\delta = (\phi_1, \dots, \phi_n)$ be a finite nonempty sequence of formulas in $\Phi_{\NXXYY}$. We use $\bigwedge \Diamond \delta$ to indicate the formula $\Diamond \phi_1 \land \dots \land \Diamond \phi_n$.

\item Let $\Theta$ be a finite set of finite non-empty sequences of formulas in $\Phi_{\NXXYY}$. We use $\bigvee \bigwedge \Diamond \Theta$ to indicate the formula $\bigvee \{\bigwedge \Diamond \delta \mid \delta \in \Theta\}$.

\ei

\begin{lemma} \label{lemma:from consistency to consistency}
Let $\CF = \fcf$ be a core formula. Then, if $\CF$ is consistent, there is an atomic sequence $\AS = \fas$ for $\CF$ such that (1) all its elements are consistent and (2) if some element of it implies $\YYY \bot$\footnote{This means $\vdash_{\OneBXXYY} \psi \rightarrow \YYY \bot$ for some element $\psi$ of $\AS$.}, then $l_1 \land \dots \land l_n$ is consistent, where $l_1, \dots, l_n$ are all the literals that are a conjunct of some element of $\AS$.
\end{lemma}

\begin{proof}
Assume $\CF$ is consistent. Let $\BS = \fbs$ be its basic sequence. We can easily show $\vdash_{\OneBXXYY} \CF \rightarrow \bigwedge \Diamond \BS$. Then, $\bigwedge \Diamond \BS$ is consistent. Note $\vdash_{\OneBXXYY} \bigwedge \Diamond \BS \leftrightarrow \bigvee \bigwedge \Diamond \{\AS \mid \text{$\AS$ is an atomic sequence for $\CF$}\}$. Then, there is an atomic sequence $\AS = \fas$ for $\CF$ such that $\bigwedge \Diamond \AS$ is consistent. Then, all the elements of $\AS$ are consistent. Assume some element of $\AS$ implies $\YYY \bot$. Let $l_1, \dots, l_n$ be all the literals that are a conjunct of some element of $\AS$. By Axiom \ref{AXXYY axiom:if Y bot}, we can get $\vdash_{\OneBXXYY} \bigwedge \Diamond \AS \rightarrow (l_1 \land \dots \land l_n)$. Then, $l_1 \land \dots \land l_n$ is consistent.
\end{proof}

We say that a contextualized pointed model $(\MM,\CC,\pi,i)$ is \emph{linear} if the domain of $\MM$ just consists of the elements of $\pi$.

\begin{lemma} \label{lemma:the bridge}
Let $\phi = \XXX^{h_1} l_1 \land \dots \land \XXX^{h_n} l_n \land \YYY^{j_1} l'_1 \land \dots \land \YYY^{j_m} l'_m$, where all $l_x$ and $l'_x$ are literals. Then, if $\phi$ is consistent, it is satisfiable at a linear contextualized pointed model.
\end{lemma}

\begin{proof}
Assume $\phi$ is consistent. It is easy to see that there is a formula $\phi'$ such that $\vdash_{\OneBXXYY} \phi \leftrightarrow \phi'$ and $\phi'$ is in the form of $\XXX^{a_1} \phi_1 \land \dots \land \XXX^{a_n} \phi_n \land \YYY^{b_1} \psi_1 \land \dots \land \YYY^{b_m} \psi_m$, where $0 \leq a_1 < \dots < a_n$, $0 < b_1 < \dots < b_m$, and $\phi_1, \dots, \phi_n, \psi_1, \dots, \psi_m$ are conjunctions of literals. By Axiom \ref{AXXYY axiom:X dual top}, $\phi_1, \dots, \phi_n$ are all consistent. Then, we can easily define a linear contextualized pointed model $(\MM,\CC,\pi,0)$ such that $\MM,\CC,\pi,0 \Vdash \XXX^{a_1} \phi_1 \land \dots \land \XXX^{a_n} \phi_n \land \YYY^{b_1} \psi_1 \land \dots \land \YYY^{b_m} \psi_m$. We skip the details. Then, it can be easily shown $\MM,\CC,\pi,0 \Vdash \phi$.
\end{proof}

\begin{lemma} \label{lemma:from satisfiability to satisfiability}
Let $\CF = \fcf$ be a core formula. Then, $\CF$ is satisfiable if there is an atomic sequence $\AS = \fas$ for $\CF$ such that (1) all its elements are satisfiable and (2) if some element of it implies $\YYY \bot$\footnote{This means $\models_{\OneBXXYY} \psi \rightarrow \YYY \bot$ for some element $\psi$ of $\AS$.}, then $l_1 \land \dots \land l_n$ is satisfiable, where $l_1, \dots, l_n$ are all the literals that are a conjunct of some element of $\AS$.
\end{lemma}

\begin{proof}
~

\bi

\item

Assume there is an atomic sequence $\AS = \fas$ for $\CF$ such that (1) all its elements are satisfiable, and (2) some element of it implies $\YYY \bot$ and $l_1 \land \dots \land l_n$ is satisfiable, where $l_1, \dots, l_n$ are all the literals that are a conjunct of some element of $\AS$.

\item

By Lemma \ref{lemma:the bridge}, we can get that there are pairwise disjoint linear contextualized pointed models $(\MM_1,\epsilon,\pi_1,0)$, $\dots$, $(\MM_i,\epsilon,\pi_i,0)$, $(\MM',\epsilon,\pi',0)$, $(\MM'',\epsilon,\pi'',0)$ which respectively satisfy the elements of $\AS$ and $l_1 \land \dots \land l_n$.

\item

Let $w$ be the root of $\MM''$. Let $\MM$ be the model constructed in the following way. Replace the roots of $\MM_1, \dots, \MM_i$, $\MM'$ by $w$. Let $\NN_1, \dots, \NN_i, \NN'$ be the resulted models and $\lambda_1, \dots, \lambda_i, \lambda'$ be the timelines of these models. Merge the resulted models in an imagined way.

\item

It can be verified that the elements of $\AS$ are respectively true at the following contextualized pointed models: $(\MM,\epsilon,\lambda_1,0), \dots, (\MM,\epsilon,\lambda_i,0), (\MM,\epsilon,\lambda',0)$.

\item

Let $\BS = \fbs$ be the basic sequence for $\CF$. Then, all its elements are respectively true at these contextualized pointed models. It can be verified $\MM,\epsilon,\lambda',0 \Vdash \CF$.

\item

Assume there is an atomic sequence $\AS = \fas$ for $\CF$ such that (1) all its elements are satisfiable and (2) no element of it implies $\YYY \bot$.

\item

By Lemma \ref{lemma:the bridge}, we can get that there are pairwise disjoint linear contextualized pointed models $(\MM_1,\epsilon,\pi_1,0)$, $\dots$, $(\MM_i,\epsilon,\pi_i,0)$, $(\MM',\epsilon,\pi',0)$ which respectively satisfy the elements of $\AS$.

\item

Let $w$ be a new state for these models. Let $\MM$ be the resulted model by merging $w$ and these models in an imagined way. Let $\lambda_1, \dots, \lambda_i, \lambda'$ be the respective results of prefixing $\pi_1, \dots, \pi_i, \pi'$ with $w$.

\item

It can be verified that the elements of $\AS$ are respectively true at the contextualized pointed models $(\MM,\epsilon,\lambda_1,0)$, $\dots$, $(\MM,\epsilon,\lambda_i,0)$, and $(\MM,\epsilon,\lambda',0)$.

\item

Let $\BS = \fbs$ be the basic sequence for $\CF$. Then, its elements are respectively true at these contextualized pointed models. It can be verified $\MM,\epsilon,\lambda',0 \Vdash \CF$.

\ei
\end{proof}

\begin{theorem}
The axiomatic system $\OneBXXYY$ is sound and complete with respect to the set of valid formulas of $\Phi_{\OneBXXYY}$.
\end{theorem}

\begin{proof}
The soundness of $\OneBXXYY$ can be easily verified and we skip it. Let $\phi$ be a consistent formula in $\Phi_{\OneBXXYY}$. Let $\phi'$ be a formula in $\Phi_{\OneBXXYY}$ such that $\vdash_{\OneBXXYY} \phi \leftrightarrow \phi'$ and $\phi'$ is in the disjunctive normal form. Then, some disjunct $\psi$ of $\phi'$ is consistent, where $\psi$ is in the form of $\Box \HX_1 \land \dots \land \Box \HX_h \land \Diamond \IX_1 \land \dots \land \Diamond \IX_i \land \LX$, where all $\HX_x, \IX_x$ and $\LX$ are in $\Phi_{\NXXYY}$. Note $\vdash_{\OneBXXYY} \Box \top$ and $\vdash_{\OneBXXYY} \top$. By Axiom \ref{AXXYY axiom:A T}, $\vdash_{\OneBXXYY} \Diamond \top$. Then, we can assume that some $\HX_x$, some $\IX_x$ and $\LX$ exist. Let $\HX = \HX_1 \land \dots \land \HX_h$. Let $\psi' = \Box \HX \land \Diamond \IX_1 \land \dots \land \Diamond \IX_i \land \LX$. It can be shown $\vdash_{\OneBXXYY} \psi \leftrightarrow \psi'$. Then, $\psi'$ is consistent. Note $\psi'$ is a core formula. By Lemmas \ref{lemma:from consistency to consistency}, \ref{lemma:the bridge} and \ref{lemma:from satisfiability to satisfiability}, $\psi'$ is satisfiable. By soundness of $\OneBXXYY$, $\phi$ is satisfiable. Then, $\OneBXXYY$ is complete.
\end{proof}

\subsection{Soundness and completeness of $\ConSHNBT$}
\label{subsec:Soundness and completeness of ConSONBT}

Let $\kappa$ be the function from $\Phi_{\ConSHNBT}$ to $\Phi_{\ConXXYY}$ mentioned in Fact \ref{fact:reduction 1 semantics}. The following result can be easily shown:

\begin{lemma}[] \label{lemma:reduction 1 system}
For every $\phi \in \Phi_{\ConSHNBT}$, $\vdash_{\ConSHNBT} \phi \leftrightarrow \kappa(\phi)$.
\end{lemma}

Let $\mu$ be the function from $\Phi_{\ConXXYY}$ to $\Phi_{\OneBXXYY}$ mentioned in Fact \ref{fact:reduction 2 semantics}. The following result can be shown similarly to how Fact \ref{fact:reduction 2 semantics} is shown.

\begin{lemma} \label{lemma:reduction 2 system}
For every $\phi \in \Phi_{\ConXXYY}$, $\vdash_{\ConSHNBT} \phi \leftrightarrow \mu(\phi)$.
\end{lemma}

\setcounter{rtheorem}{0}

\begin{rtheorem}
The axiomatic system $\ConSHNBT$ is sound and complete with respect to the set of valid formulas of $\Phi_{\ConSHNBT}$.
\end{rtheorem}

\begin{proof}
The soundness of $\ConSHNBT$ can be easily verified and we skip it. Let $\phi \in \Phi_{\ConSHNBT}$ such that $\models_{\ConSHNBT} \phi$. By Facts \ref{fact:reduction 1 semantics} and \ref{fact:reduction 2 semantics}, $\models_{\ConSHNBT} \mu(\kappa(\phi))$. Note $\mu(\kappa(\phi)) \in \Phi_{\OneBXXYY}$. By completeness of $\OneBXXYY$, $\vdash_{\OneBXXYY} \mu(\kappa(\phi))$. Note $\OneBXXYY$ is contained in $\ConSHNBT$. Then, $\vdash_{\ConSHNBT} \mu(\kappa(\phi))$. By Lemmas \ref{lemma:reduction 1 system} and \ref{lemma:reduction 2 system}, $\vdash_{\ConSHNBT} \phi$. Then, $\ConSHNBT$ is complete.
\end{proof}

\end{document}